\documentclass{article}%
\usepackage{amsmath}
\usepackage{amsfonts}
\usepackage{amssymb}
\usepackage{graphicx}%
\providecommand{\U}[1]{\protect\rule{.1in}{.1in}}
\newtheorem{theorem}{Theorem}[section]

\newtheorem{corollary}[theorem]{Corollary}

\newtheorem{definition}[theorem]{Definition}
\newtheorem{example}[theorem]{Example}

\newtheorem{lemma}[theorem]{Lemma}

\newtheorem{proposition}[theorem]{Proposition}
\newtheorem{remark}[theorem]{Remark}

\newenvironment{proof}[1][Proof]{\noindent \textbf{#1.} }{\  \rule{0.5em}{0.5em}}
\numberwithin{equation}{section}
\begin{document}

\title{The Pricing Mechanism \\of Contingent Claims and its Generating Function}
\author{Shige PENG\thanks{The author thanks the partial support from the Natural
Science Foundation of China, grant No. 10131040. This research is supported in
part by The National Natural Science Foundation of China No. 10131040. This
reversion is made after the author's visit, during November 2003, to Institute
of Mathematics and System Science, Academica Sinica, where he gives a series
of lectures on this paper. He thanks Zhiming Ma and Jia-an Yan, for their
fruitful suggestions, critics and warm encouragements. He also thanks to
Claude Dellacherie for his suggestions and critics. }\\Institute of Mathematics, Institute of Finance\\Shandong University\\250100, Jinan, China\\peng@sdu.edu.cn}
\date{Version: March 31, 2006}
\maketitle
\tableofcontents

\bigskip\medskip\medskip

\noindent\textbf{Abstract. }{\small In this paper we study dynamic pricing
mechanism of contingent claims. A typical model of such pricing mechanism is
the so-called }$g${\small --expectation }${\small E}_{s,t}^{g}{\small [X]}$
{\small defined by the solution of the backward stochastic differential
equation with generator }$g$ {\small and with the contingent claim
}${\small X}$ {\small as terminal condition. The generating function }$g$
{\small this BSDE. We also provide examples of determining the price
generating function }${\small g=g(y,z)}$ {\small by testing. }

{\small The main result of this paper is as follows: if a given dynamic
pricing mechanism is $\mathbb{E}^{g_{\mu}}$--dominated, i.e., the criteria
(A5) $\mathbb{E}_{s,t}[X]-\mathbb{E}_{s,t}[X^{\prime}]\leq\mathbb{E}^{g_{\mu}%
}[X-X^{\prime}]$ is satisfied for a large enough $\mu>0$, where $g_{\mu}%
=\mu(|y|+|z|)$, then $\mathbb{E}_{s.t}[\cdot]$ is a $g$--pricing mechanism.
This domination condition was statistically tested using CME data docoments.
The result of test is significantly positive.}

\medskip

\medskip

\noindent\textbf{MSC 2000 Classification Numbers: }60H10, 60H05,
\newline\newline

\section{Introduction\label{ss1}\bigskip}

There are a lot of data of the processes of prices of huge variety of
contingent claims, vanilla options, exotic options, etc. Each process
corresponds the price of a specific contingent claim issued in a specific
market and offered by a specific financial institution. A typical example is
the call and put options with a specific stock price as their underlying
asset. We can find the real time data of of the option price $C_{t}$, $t\geq
t_{0}$ for a call option $C_{T}=(S_{T}-k)^{+}$ with $T$ as its maturity. There
exist many processes of prices of this specific product, e.g., the bid price,
the ask price, the we-buy price and we-sell price by a market maker under a
specific background, etc. The main point of view of this paper is, behind a
price process, there is a pricing mechanism. Take the above option market
price $C_{t}$ for example, there exists a mapping $\mathbb{E}_{t,T}[\cdot]$
from $\Lambda_{T}$ the space of option price states at time $T$ to
$\Lambda_{t}$ at the time $t\in\lbrack t_{0},T]$ such that $C_{t}$ is produced
by $\mathbb{E}_{t,T}[C_{T}]$. This family of mapping
\[
\mathbb{E}_{t,T}[X]:X\in\Lambda_{T}\longmapsto\Lambda_{t},\ t\leq T
\]
forms the pricing mechanism for this specific option market prices.

Black-Scholes formula can be regarded as a dynamic pricing mechanism
of contingent claim. In fact, it can be regarded as to solve a
specific linear backward stochastic differential equation (BSDE).
More generally, each BSDE with a given generating function $g$ forms
a dynamic model of pricing mechanism of contingent claims.

In this paper we explain the following result: if an a dynamic
pricing mechanism is dominated by $g_{\mu}$--pricing mechanism, with
large enough $\mu>0$, then it is a $g$--pricing mechanism: there
exists a unique generating function $g$, such that the price of the
pricing mechanism is solved by the corresponding BSDE. In this case,
to find the corresponding generating function $g$ by using data of
the pricing process is a very interesting problem, since $g$
determines entirely the pricing mechanism. The domination condition
can be tested also by data analysis of the price processes.

The paper is organized as follows: in section \ref{ss2}, we present the the
notion of $\mathcal{F}_{t}$--consistent pricing mechanisms in subsection
\ref{ss2.1}. We then give a concrete $\mathcal{F}_{t}$--pricing mechanism:
$\mathbb{E}^{g}$--pricing mechanisms in subsection \ref{ss2.2}. The main
result, Theorem \ref{m7.1}, will be presented in section \ref{ss3}. We also
provide some examples and explain how to find the function $g$ through by
testing the input--output data. This main theorem will be proved in Section
\ref{ss8}. Nonlinear decomposition theorems of Doob--Meyer's type, i.e.,
Proposition \ref{p2.3} and Theorem \ref{m6.1} play crucial roles in the proof
of Theorem \ref{m7.1}. Theorem \ref{m6.1} has also an interesting
interpretation in finance (see Remark \ref{m6.1Rem1}).

The crucial domination inequality (\ref{e3.1}) of our main result Theorem
\ref{m7.1} is tested by using data of parameter files, provided by CME, for
options based on S\&P500 futures. The result strongly support that the option
pricing mechanism of CME is a $g$--pricing mechanism.

Another application of the dynamical expectations and pricing mechanisms is to
risk measures. Axiomatic conditions for a (one step) coherent risk measure was
introduced by Artzner, Delbaen, Eber and Heath 1999 \cite{ADEH1999} and, for a
convex risk measure, by F\"{o}llmer and Schied (2002) {\cite{Fo-Sc}}. Rosazza
Gianin (2003) studied dynamical risk measures using the notion of
$g$--expectations in \cite{rosazza} (see also \cite{Peng2003b}, \cite{El-Bar},
\cite{El-Bar2005}) in which (B1)--(B4) are satisfied. In fact conditions
(A1)-(A4), as well as their special situation (B1)--(B4) provides an ideal
characterization of the dynamical behaviors of a the a risk measure. But in
this paper we emphasis the study of the mechanism of the pricing mechanism to
a further payoff, for which is, in general, the translation property in risk
measure is not satisfied.

\section{The pricing mechanisms and $g$--pricing mechanism by BSDE\label{ss2}}

\subsection{Basic setting\label{ss2.1}}

We assume that the price $S$ of the underlying assets is driven by a
$d$--dimensional Brownian motion $(B_{t})_{t\geq0}$ in a probability space
$(\Omega,\mathcal{F},P)$. We don't need to precise the model of $S_{t}$, what
we assume here is that the information $\mathcal{F}_{t}^{S}$ of the price $S$
coincides with that of the Brownian motion:
\[
\mathcal{F}_{t}^{S}=\mathcal{F}_{t}:=\sigma\{B_{s},\;s\leq t\}
\]
For each $t\in\lbrack0,\infty)$, the state of contingent prices will be given
in the following space

\begin{itemize}
\item $\Lambda_{t}=L^{2}(\mathcal{F}_{t}):=$\{the space of all real valued
$\mathcal{F}_{t}$--measurable random variables such that $E[|X|^{p}]<\infty$\}.
\end{itemize}

\begin{definition}
\label{d2.1}A system of operators:
\[
\mathbb{E}_{s,t}[X]:X\in L^{2}(\mathcal{F}_{t})\rightarrow L^{2}%
(\mathcal{F}_{s}),\;T_{0}\leq s\leq t\leq T_{1}%
\]
is called an $\mathcal{F}_{t}$--\textrm{consistent pricing mechanism} defined
on $[T_{0},T_{1}]$ if it satisfies the following properties: for each
$T_{0}\leq s\leq t\leq T_{1}$ and for each $X_{t}$, $X_{t}^{\prime}\in
L^{2}(\mathcal{F}_{t})$, \newline\textbf{(A1)} $\mathbb{E}_{s,t}[X_{t}%
]\geq\mathbb{E}_{s,t}[X_{t}^{\prime}]$, a.s., if $X_{t}\geq X_{t}^{\prime}$,
a.s.; \newline\textbf{(A2)} $\mathbb{E}_{t,t}[X_{t}]=X_{t}$, a.s.;
\newline\textbf{(A3)} $\mathbb{E}_{r,s}[\mathbb{E}_{s,t}[X_{t}]]=\mathbb{E}%
_{r,t}[X_{t}]$, a.s.; for $r\leq s$\newline\textbf{(A4)} $1_{A}\mathbb{E}%
_{s,t}[X_{t}]=1_{A}\mathbb{E}_{s,t}[1_{A}X_{t}]$, a.s.\ $\forall
A\in\mathcal{F}_{s}$.
\end{definition}

We will often consider (A1)--(A4) plus an additional condition: \newline%
\textbf{(A4}$_{0}$\textbf{)} $\mathbb{E}_{s,t}[0]=0$, a.s.\ $\forall0\leq
s\leq t\leq T$.

\begin{remark}
The raison we use the letter $\mathbb{E}_{s,t}[\cdot]$ to denote the above
pricing mechanism is that its behavior is very like the conditional
expectation $E[X_{t}|\mathcal{F}_{s}]$ for a $\mathcal{F}_{t}$--measurable
random variable. It is wise profit this similarity to introduce the notion of
$\mathbb{E}$-martingales which are the data of the processes of option prices
produced by this pricing mechanism.
\end{remark}

\begin{remark}
(A1) and (A2) are economically obvious conditions. Condition (A3) means that
the value $\mathbb{E}_{s,t}[X_{t}]$ can be regarded as a contingent claim at
the maturity $s$. The price of this contingent claim at the time $r\leq s$ is
$\mathbb{E}_{r,s}[\mathbb{E}_{s,t}[X_{t}]]$. It have to be the same as the
price $\mathbb{E}_{r,t}[X_{t}]$.
\end{remark}

\begin{remark}
The meaning of condition (A4) is: at time $s$, the agent knows whether $I_{A}$
worthes $1$. If it is $1$, then the price $\mathbb{E}_{s,t}[1_{A}X_{t}]$ must
be the same as $\mathbb{E}_{s,t}[X_{t}]$.
\end{remark}

\begin{proposition}
\label{m2a4}(A4) plus (A4$_{0}$) is equivalent to \newline\textbf{(A4')}
$1_{A}\mathbb{E}_{s,t}[X]=\mathbb{E}_{s,t}[1_{A}X]$, a.s.\ $\forall
A\in\mathcal{F}_{s}$.
\end{proposition}

\begin{proof}
\textbf{. } It is clear that (A4') implies (A4). $\mathbb{E}_{s,t}[0]\equiv0$
can be derived by putting $A=\emptyset$ in (A4'). On the other hand, (A4) plus
the additional condition implies
\[
1_{A^{C}}\mathbb{E}_{s,t}[1_{A}X]=1_{A^{C}}\mathbb{E}_{s,t}[1_{A^{c}}%
1_{A}X]=0.
\]
We thus have
\begin{align*}
\mathbb{E}_{s,t}[1_{A}X]  &  =1_{A^{C}}1_{A}\mathbb{E}_{s,t}[X]+1_{A}%
1_{A}\mathbb{E}_{s,t}[X]\\
&  =1_{A}\mathbb{E}_{s,t}[X].
\end{align*}

\end{proof}

\begin{proposition}
\label{mA4} (A4) is equivalent to, for each $0\leq s\leq t$ and $X,X^{\prime
}\in L^{2}(\mathcal{F}_{t})$,
\begin{equation}
\mathbb{E}_{s,t}[1_{A}X+1_{A^{C}}X^{\prime}]=1_{A}\mathbb{E}_{s,t}%
[X]+1_{A^{C}}\mathbb{E}_{s,t}[X^{\prime}],\; \mathrm{a.s.}\; \forall
A\in\mathcal{F}_{s}. \label{eA4}%
\end{equation}

\end{proposition}

\begin{proof}
(A4) $\Rightarrow$ (\ref{eA4}): We let $Y=1_{A}X+1_{A^{C}}X^{\prime}$. Then,
by (A4)
\[
1_{A}\mathbb{E}_{s,t}[Y]=1_{A}\mathbb{E}_{s,t}[1_{A}Y]=1_{A}\mathbb{E}%
_{s,t}[1_{A}X]=1_{A}\mathbb{E}_{s,t}[X].
\]
Similarly
\[
1_{A^{C}}\mathbb{E}_{s,t}[Y]=1_{A^{C}}\mathbb{E}_{s,t}[1_{A^{C}}Y]=1_{A^{C}%
}\mathbb{E}_{s,t}[1_{A^{C}}X^{\prime}]=1_{A^{C}}\mathbb{E}_{s,t}[X^{\prime}].
\]
Thus (\ref{eA4}) from $1_{A}\mathbb{E}_{s,t}[Y]+1_{A^{C}}\mathbb{E}%
_{s,t}[Y]=1_{A}\mathbb{E}_{s,t}[X]+1_{A^{C}}\mathbb{E}_{s,t}[X^{\prime}]$.
\end{proof}

(\ref{eA4}) $\Rightarrow$ (A4): It is simply because of
\begin{align*}
1_{A}\mathbb{E}_{s,t}[1_{A}X]  &  =1_{A}\mathbb{E}_{s,t}[1_{A}X+1_{A^{C}%
}(1_{A}X)]\\
&  =1_{A}(1_{A}\mathbb{E}_{s,t}[X]+1_{A^{C}}\mathbb{E}_{s,t}[1_{A^{C}}X])\\
&  =1_{A}\mathbb{E}_{s,t}[X].
\end{align*}

\begin{remark}
At time $t$, the agent knows the value of $1_{A}$. (A4) means that, if,
$\omega\in A$, i.e.., $1_{A}(\omega)=1$ then the value $\mathbb{E}_{s,t}%
[1_{A}X]$ should be the same as $\mathbb{E}_{s,t}[X]$ since the two outcomes
$X(\omega)$ and $(1_{A}X)(\omega)$ are exactly the same. (A4) is applied to
the pricing mechanism of a final outcome $X$ plus some \textquotedblleft
dividend\textquotedblright\ $(D_{s})_{s\geq0}$.
\end{remark}

An immediate property of this dynamical pricing mechanism is that they can be
pasted together, one after the other to form a new dynamical pricing mechanism:

\begin{proposition}
\label{m2.3}Let $T_{0}<T_{1}<T_{2}<\cdots<T_{N}$ be given and, for
$i=0,1,2,\cdots,N-1$, let $\mathbb{E}_{s,t}^{i}[\cdot]$, $T_{i}\leq s\leq
t\leq T_{i+1}$ be an $\mathcal{F}_{t}$--consistent pricing mechanism on
$[T_{i},T_{i+1}]$ in the sense of Definition \ref{d2.1}. Then there exists a
unique $\mathcal{F}_{t}$--consistent pricing mechanism $\mathbb{E}[\cdot]$
defined on $[T_{0},T_{N}]$
\[
\mathbb{E}_{s,t}[X]:X\in L^{2}(\mathcal{F}_{t})\rightarrow L^{2}%
(\mathcal{F}_{s}),\;T_{0}\leq s\leq t\leq T_{N}%
\]
such that, for each $i=0,1,\cdots,N-1$, and for each $T_{i}\leq s\leq t\leq
T_{i+1}$,
\begin{equation}
\mathbb{E}_{s,t}[X]=\mathbb{E}_{s,t}^{i}[X],\; \forall X\in L^{2}%
(\mathcal{F}_{t}). \label{e2.01}%
\end{equation}

\end{proposition}

\begin{proof}
\textbf{ }It suffices to prove the case $N=2$, since we then can apply this
result to the cases $[T_{0},T_{3}]=[T_{0},T_{2}]\cup\lbrack T_{2},T_{3}]$,
$\cdots$ and finally $[T_{0},T_{N}]=[T_{0},T_{N-1}]\cup\lbrack T_{N-1},T_{N}%
]$. \newline We define
\begin{equation}
\mathbb{E}_{s,t}[X]=\left\{
\begin{array}
[c]{cll}
& \mathbb{E}_{s,t}^{1}[X]\ \ \  & T_{0}\leq s\leq t\leq T_{1};\\
& \mathbb{E}_{s,t}^{2}[X], & T_{1}\leq s\leq t\leq T_{2};\\
& \mathbb{E}_{s,T_{1}}^{1}[\mathbb{E}_{T_{1},t}^{2}[X]] & T_{1}\leq
s<T_{1}<t\leq T_{2}.
\end{array}
\right.  \label{e2.02}%
\end{equation}
It is clear that, on $[T_{0},T_{2}]$, $\mathbb{E}_{s,t}[\cdot]$ satisfies (A1)
and (A2). To prove (A3) it suffices to check the relation
\[
\mathbb{E}_{r,s}[\mathbb{E}_{s,t}[X]]=\mathbb{E}_{r,t}[X],\;T_{0}\leq r\leq
s\leq t\leq T_{1}%
\]
for two cases: $T_{0}\leq r\leq s\leq T_{1}\leq t\leq T_{2}$ and $T_{0}\leq
r\leq T_{1}\leq s\leq t\leq T_{2}$. For the first case
\begin{align*}
\mathbb{E}_{r,s}[\mathbb{E}_{s,t}[X]]  &  =\mathbb{E}_{r,s}^{1}[\mathbb{E}%
_{s,T_{1}}^{1}[\mathbb{E}_{T_{1},t}^{2}[X]]]\\
&  =\mathbb{E}_{r,T_{1}}^{1}[\mathbb{E}_{T_{1},t}^{2}[X]]\\
&  =\mathbb{E}_{r,t}[X].
\end{align*}
For the second case
\begin{align*}
\mathbb{E}_{r,s}[\mathbb{E}_{s,t}[X]]  &  =\mathbb{E}_{r,T_{1}}^{1}%
[\mathbb{E}_{T_{1},s}^{2}[\mathbb{E}_{s,t}^{2}[X]]]\\
\  &  =\mathbb{E}_{r,T_{1}}^{1}[\mathbb{E}_{T_{1},t}^{2}[X]]\\
\  &  =\mathbb{E}_{r,t}[X].
\end{align*}
We now prove (A4). Again it suffices to check the case $T_{0}\leq s\leq
T_{1}\leq t\leq T_{2}$. In this case, for each $A\in\mathcal{F}_{s}%
\subset\mathcal{F}_{T_{1}}$, (A4) is derived from
\begin{align*}
1_{A}\mathbb{E}_{s,t}[X]  &  =1_{A}\mathbb{E}_{s,T_{1}}^{1}[\mathbb{E}%
_{T_{1},t}^{2}[X]]\\
&  =1_{A}\mathbb{E}_{s,T_{1}}^{1}[1_{A}\mathbb{E}_{T_{1},t}^{2}[X]]\\
&  =1_{A}\mathbb{E}_{s,T_{1}}^{1}[\mathbb{E}_{T_{1},t}^{2}[1_{A}X]]\\
&  =1_{A}\mathbb{E}_{s,t}[1_{A}X].
\end{align*}
It remains to prove the uniqueness of $\mathbb{E}[\cdot]$. Let $\mathbb{E}%
^{a}[\cdot]$ be an $\mathcal{F}_{t}$--consistent pricing mechanism such
that,\
\[
\mathbb{E}_{s,t}^{a}[X]=\mathbb{E}_{s,t}^{i}[X],\; \forall X\in L^{2}%
(\mathcal{F}_{t}),\;i=1,2.
\]
We then have, when $T_{0}\leq s\leq t\leq T_{1}$ and $T_{1}\leq s\leq t\leq
T_{2}$, $\mathbb{E}_{s,t}^{a}[X]\equiv\mathbb{E}_{s,t}[X],\; \forall X\in
L^{2}(\mathcal{F}_{t})$. For the remaining case, i.e., $T_{0}\leq
s<T_{1}<t\leq T_{1}$, since $\mathbb{E}^{a}$ satisfies (A3),
\begin{align*}
\mathbb{E}_{s,t}^{a}[X]  &  =\mathbb{E}_{s,T_{1}}^{a}[\mathbb{E}_{T_{1},t}%
^{a}[X]]\\
&  =\mathbb{E}_{s,T_{1}}^{1}[\mathbb{E}_{T_{1},t}^{2}[X]]\\
&  =\mathbb{E}_{s,t}[X],\; \forall X\in L^{2}(\mathcal{F}_{t}).
\end{align*}
Thus $\mathbb{E}_{s,t}^{a}[\cdot]=\mathbb{E}_{s,t}[\cdot]$. This completes the proof.\
\end{proof}

\subsection{Dynamic pricing mechanism generated by BSDE\label{ss2.2}}

We need the following notations. Let $p\geq1$ and $t\in\lbrack0,\infty)$ be given.

\medskip

\begin{itemize}
\item $L^{p}(\mathcal{F}_{t};R^{m}):=$\{the space of all $R^{m}$--valued
$\mathcal{F}_{t}$--measurable random variables such that $E[|\xi|^{p}]<\infty$\};

\item $L_{\mathcal{F}}^{p}(0,t;R^{m}):=$\{$R^{m}$--valued and predictable
stochastic processes such that $E\int_{0}^{t}|\phi_{s}|^{p}ds<\infty$\};

\item $D_{\mathcal{F}}^{p}(0,t;R^{m}):=$\{all RCLL processes in
$L_{\mathcal{F}}^{p}(0,t;R^{m})$ such that $E[\sup_{0\leq s\leq t}|\phi
_{s}|^{p}]<\infty$\};

\item $S_{\mathcal{F}}^{p}(0,T;R^{m}):=$\{all continuous processes in
$D_{\mathcal{F}}^{p}(0,T;R^{m})$ \};
\end{itemize}

In the case $m=1$, we denote them by $L^{p}(\mathcal{F}_{t})$, $L_{\mathcal{F}%
}^{p}(0,t)$, $D_{\mathcal{F}}^{p}(0,t)$ and $S_{\mathcal{F}}^{p}(0,t)$. We
recall that all elements in $D_{\mathcal{F}}^{2}(0,T)$ are predictable.

\medskip

For each given $X\in L^{2}(\mathcal{F}_{t})$, we solve the following BSDE on
$[0,t]$:%

\begin{equation}
Y_{s}=X+\int_{s}^{t}g(r,Y_{r},Z_{r})dr-\int_{s}^{t}Z_{r}dB_{r}, \label{tsBSDE}%
\end{equation}
where the unknown is the pair of the adapted processes $(Y,Z)$. Here the
function $g$ is given
\[
g:(\omega,t,y,z)\in\Omega\times\lbrack0,\infty)\times R\times R^{d}\rightarrow
R.
\]
It satisfies the following basic assumptions for each $\forall y,y^{\prime}\in
R,\;z,z^{\prime}\in R^{d}$
\begin{equation}
\left\{
\begin{array}
[c]{rrl}
&  & g(\cdot,y,z)\in L_{\mathcal{F}}^{2}(0,T),\ \ \forall T\in(0,\infty),\\
&  & |g(t,y,z)-g(t,y^{\prime},z^{\prime})|\leq\mu(|y-y^{\prime}|+|z-z^{\prime
}|)\;.
\end{array}
\right.  \label{h2.1}%
\end{equation}
In some cases it is interesting to consider the following situation:
\begin{equation}
\left\{
\begin{array}
[c]{rrl}%
\text{(a)\ \ \ } & g(\cdot,0,0) & \equiv0,\\
\text{(b)\ \ \ } & g(\cdot,y,0) & \equiv0,\; \forall y\in R.
\end{array}
\ \right.  \label{h2.2}%
\end{equation}
Obviously (b) implies (a). This BSDE (\ref{tsBSDE}) was intrduced by Bismut
\cite{Bismut73}, \cite{Bismut78} for the case where $g$ is a linear function
of $(y,z)$. Pardoux and Peng \cite{Pardoux-Peng1990} obtained the following
result (see Theorem \ref{th2.1} for a more general situation): for each $X\in
L^{2}(\mathcal{F}_{t})$, there exists a unique solution $(Y,Z)\in
S_{\mathcal{F}}^{2}(0,t)\times L_{\mathcal{F}}^{2}(0,t;R^{d})$ of the BSDE
(\ref{tsBSDE}).

\begin{definition}
\label{EgstX}We denote by $\mathbb{E}_{s,t}^{g}[X_{t}]:=Y_{s}$, $0\leq s\leq
t$. We thus define a system of operators
\begin{equation}
\mathbb{E}_{s,t}^{g}[\cdot]:L^{2}(\mathcal{F}_{t})\rightarrow L^{2}%
(\mathcal{F}_{s}),\ \ \ 0\leq s\leq t<\infty. \label{Def2.2aa}%
\end{equation}
$(\mathbb{E}_{s,t}^{g}[\cdot])_{0\leq s\leq t<\infty}$ is called $g$--expectation.
\end{definition}

\begin{proposition}
\label{p2.1aa}Let the generating function $g$ satisfies (\ref{h2.1}). Then
\[
\mathbb{E}_{s,t}^{g}[X]:X\in L^{2}(\mathcal{F}_{t})\rightarrow L^{2}%
(\mathcal{F}_{s}),\;0\leq s\leq t<\infty
\]
defined in (\ref{Def2.2aa}) is an $\mathcal{F}_{t}$--consistent pricing
mechanism, called $g$--pricing mechanism, i.e., it satisfies (A1)--(A4) of
Definition \ref{d2.1}.
\end{proposition}

This pricing mechanism is entirely generated by function $g$. We then call $g$
a (contingent claim) price generating function.

\begin{proof}
This result is a special case of Proposition \ref{p2.1}.
\end{proof}

Since $g$ satisfies Lipschitz condition with Lipschitz constant $\mu$, it is
then dominated by the following function
\begin{equation}
g_{\mu}(y,z):=\mu|y|+\mu|z|,\;(y,z)\in R\times R^{d} \label{e2.9}%
\end{equation}
in the following since%
\[
g(t,y,z)-g(t,y^{\prime},z^{\prime})\leq g_{\mu}(y-y^{\prime},z-z^{\prime}).
\]
We will see that the above notion of domination is useful. Briefly speaking, a
price generating function $g$ is dominated by another one if and only if the
corresponding pricing mechanism $\mathbb{E}^{g}$ is dominated by the other one.

\section{Main result: $\mathbb{E}_{s,t}[\cdot]$ is governed by a BSDE
\label{ss3}}

From now on the system $\mathbb{E}_{s,t}[\cdot]_{0\leq s\leq t<\infty}$ is
always a fixed $\mathcal{F}_{t}$--consistent pricing mechanism, i.e.,
satisfying (A1)--(A4), with additional assumptions (A4$_{0}$) and the
following ${\mathbb{E}}^{g_{\mu}}$--domination assumption:

\smallskip\

\textbf{(A5)} there exists a sufficiently large number $\mu>0$ such that, for
each $0\leq s\leq t\leq T$,
\begin{equation}
\mathbb{E}_{s,t}[X]-\mathbb{E}_{s,t}[X^{\prime}]\leq\mathbb{E}_{s,t}^{g_{\mu}%
}[X-X^{\prime}],\; \; \forall X,X^{\prime}\in L^{2}(\mathcal{F}_{t}),
\label{e3.1}%
\end{equation}
where the function $g_{\mu}(y,z)=\mu|y|+\mu|z|$ is given in (\ref{e2.9}%
).\medskip

The \textbf{main theorem} of this paper is:

\begin{theorem}
\label{m7.1}We assume that the function $g$ satisfies (\ref{h2.1}) with
$g(\cdot,0,0)=0$. Then the $g$--expectation $\mathbb{E}_{s,t}^{g}%
[\cdot]_{0\leq s\leq t<\infty}$ is an $\mathcal{F}_{t}$--consistent pricing
mechanism satisfying (A1)--(A4), (A4$_{0}$) and the domination condition (A5).
$\mathbb{E}^{g}$ is then called $g$--(contingent claim) pricing mechanism, and
the function $g$ is called a (contingent claim) price generating
function.\newline Conversely, let $\mathbb{E}_{s,t}[\cdot]_{0\leq s\leq
t<\infty}$ be an $\mathcal{F}_{t}$--consistent pricing mechanism satisfying
(A1)--(A4), (A4$_{0}$) and the domination condition (A5), then there exists a
unique price generating function $g(\omega,t,y,z)$ satisfying (\ref{h2.1})
with $g(\cdot,0,0)\equiv0$, such that
\begin{equation}
\mathbb{E}_{s,t}[X]=\mathbb{E}_{s,t}^{g}[X],\; \; \forall s\leq t,\ \forall
X\in L^{2}(\mathcal{F}_{t}). \label{e7.1}%
\end{equation}

\end{theorem}

\begin{remark}
The case where $\mathbb{E}_{s,t}[\cdot]$ satisfy (A1)--(A5), without (A4$_{0}%
$), can be obtained as corollaries of the this main theorem. In this more
general situation the condition $g(s,0,0)\equiv0$ is not imposed. The main
result of \cite{CHMP2002}
\end{remark}

We consider some special situations of our theorem.

\begin{example}
If moreover, $g(s,y,0)\equiv0$. Then, by \cite{Peng1997}, (A2') holds. Thus,
according to Proposition \ref{m2a4}, $\mathbb{E}_{s,t}^{g}[\cdot]$ becomes an
$\mathcal{F}_{t}$--consistent nonlinear expectation:
\[
\mathbb{E}[X|\mathcal{F}_{t}]=\mathbb{E}_{g}[X|\mathcal{F}_{t}]:=\mathbb{E}%
_{s,t}^{g}[X]=\mathbb{E}_{s,T}^{g}[X].
\]
This is so called $g$--expectation introduced in \cite{Peng1997}.
\end{example}

This extends non trivially the result obtained in \cite{CHMP2002}, (see also
\cite{Peng2003b} for a more systematical presentation and explanations in
finance), where we needed a more strict domination condition plus the
following assumption
\[
\mathbb{E}[X+\eta|\mathcal{F}_{t}]=\mathbb{E}[X|\mathcal{F}_{t}]+\eta,\;
\forall\eta\in\mathcal{F}_{t}\hbox{.}
\]
Under these assumptions we have proved in \cite{CHMP2002} that there exists a
unique function $g=g(s,z)$, with $g(s,0)\equiv0$, such that $\mathbb{E}%
_{g}[X]\equiv\mathbb{E}[X]=\mathbb{E}[X|\mathcal{F}_{0}]$.

\begin{example}
Consider a financial market consisting of $d+1$ assets: one bond and $d$
stocks. We denote by $P_{0}(t)$ the price of the bond and by $P_{i}(t)$ the
price of the $i$-th stock at time $t$. We assume that $P_{0}$ is the solution
of the ordinary differential equation: $dP_{0}(t)=r(t)P_{0}(t)dt,$ and
$\{P_{i}\}_{i=1}^{d}$ is the solution of the following SDE
\begin{align*}
dP_{i}(t)  &  =P_{i}(t)[b_{i}(t)dt+{\sum}_{j=1}^{d}\sigma_{ij}(t)dB_{t}%
^{j}],\\
P_{i}(0)  &  =p_{i},\quad i=1,\cdots,d.
\end{align*}
Here $r$ is the interest rate of the bond; $\{b_{i}\}_{i=1}^{d}$ is the rate
of the expected return, $\{ \sigma_{ij}\}_{i,j=1}^{d}$ the volatility of the
stocks. We assume that $r$, $b$, $\sigma$ and $\sigma^{-1}$ are all
$\mathcal{F}_{t}$--adapted and uniformly bounded processes on $[0,\infty)$.
Black and Scholes have solved the problem of the market pricing mechanism of
an European type of derivative $X\in L^{2}(\mathcal{F}_{T})$ with maturity
$T$. In the point of view of BSDE, the problem can be treated as follows:
consider an investor who has, at a time $t\leq T$, $n_{0}(t)$ bonds and
$n_{i}(t)$ $i$-stocks, $i=1,\cdots,d$, i.e., he invests $n_{0}(t)P_{0}(t)$ in
bond and $\pi_{i}(t)=n_{i}(t)P_{i}(t)$ in the $i$-th stock. $\pi(t)=(\pi
_{1}(t),\cdots,\pi_{d}(t))$, $0\leq t\leq T$ is an $R^{d}$ valued,
square-integrable and adapted process. We define by $y(t)$ the investor's
wealth invested in the market at time $t$:
\[
y(t)=n_{0}(t)P_{0}(t)+{\sum}_{i=1}^{d}\pi_{i}(t).
\]
We make the so called self--financing assumption: in the period $[0,T]$, the
investor does not withdraw his money from, or put his money in his account
$y_{t}$. Under this condition, his wealth $y(t)$ evolves according to
\[
dy(t)=n_{0}(t)dP_{0}(t)+{\sum}_{i=1}^{d}n_{i}(t)dP_{i}(t).
\]
or
\[
dy(t)=[r(t)y(t)+{\sum}_{i=1}^{d}(b_{i}(t)-r(t))\pi_{i}(t)]dt+{\sum}%
_{i,j=1}^{d}\sigma_{ij}(t)\pi_{i}(t)dB_{t}^{j}.
\]
We denote $g(t,y,z):=-r(t)y-{\sum}_{i,j=1}^{d}(b_{i}(t)-r(t))\sigma_{ij}%
^{-1}(t)z_{j}$. Then, by the variable change $z_{j}(t)={\sum}_{i=1}^{d}%
\sigma_{ij}(t)\pi_{i}(t)$, the above equation is
\[
-dy(t)=g(t,y(t),z(t))dt-z(t)dB_{t}.
\]
We observe that function $g$ satisfies (\ref{h2.1}). It follows from the
existence and uniqueness theorem of BSDE (Theorem \ref{th2.1}) that for each
derivative $X\in L^{2}(\mathcal{F}_{T})$, there exists a unique solution
$(y(\cdot),z(\cdot))\in L_{\mathcal{F}}^{2}(0,T;R^{1+d})$ with the terminal
condition $y_{T}=X$. This meaning is significant: in order to replicate the
derivative $X$, the investor needs and only needs to invest $y(t)$ at the
present time $t$ and then, during the time interval $[t,T]$ and then to
perform the portfolio strategy $\pi_{i}(s)=\sigma_{ij}^{-1}(s)z_{j}(s)$.
Furthermore, by Comparison Theorem of BSDE, if he wants to replicate a
$X^{\prime}$ which is bigger than $X$, (i.e., $X^{\prime}\geq X$, a.s.,
$P(X^{\prime}\geq X)>0$), then he must pay more, i.e., there this no arbitrage
opportunity. This $y(t)$ is called the Black--Scholes price, or Black--Scholes
pricing mechanism, of $X$ at the time $t$. We define, as in (\ref{e2.4}),
$\mathbb{E}_{t,T}^{g}[X]=y_{t}$. We observe that the function $g$ satisfies
(b) of condition (\ref{h2.2}). It follows from Proposition \ref{p2.1} that
$\mathbb{E}_{t,T}^{g}[\cdot]$ satisfies properties (A1)--(A4) for
$\mathcal{F}_{t}$--consistent pricing mechanism.
\end{example}

\begin{example}
An very important problem is: if we know that the pricing mechanism of an
investigated agent is a $g$--pricing mechanism $\mathbb{E}^{g}$, how to find
this price generating function $g$. We now consider a case where $g$ depends
only on $z$, i.e., $g=g(z):\mathbf{R}^{d}\rightarrow\mathbf{R}$. In this case
we can find such $g$ by the following testing method. Let $\bar{z}%
\in\mathbf{R}^{d}$ be given. We denote $Y_{s}:=\mathbb{E}_{s,T}^{g}[\bar
{z}(B_{T}-B_{t})]$, $s\in\lbrack t,T]$, where $t$ is the present time. It is
the solution of the following BSDE
\[
Y_{s}=\bar{z}(B_{T}-B_{t})+\int_{s}^{T}g(Z_{u})du-\int_{s}^{T}Z_{u}%
dB_{u},\;s\in\lbrack t,T].
\]
It is seen that the solution is $Y_{s}=\bar{z}(B_{s}-B_{t})+\int_{s}^{T}%
g(\bar{z})ds$, $Z_{s}\equiv\bar{z}$. Thus
\[
\mathbb{E}_{t,T}^{g}[\bar{z}(B_{T}-B_{t})]=Y_{t}=g(\bar{z})(T-t),
\]
or
\begin{equation}
g(\bar{z})=(T-t)^{-1}\mathbb{E}_{t,T}^{g}[\bar{z}(B_{T}-B_{t})].
\label{e2.exm1}%
\end{equation}
Thus the function $g$ can be tested as follows: at the present time $t$, we
ask the investigated agent to evaluate $\bar{z}(B_{T}-B_{t})$. We thus get
$\mathbb{E}_{t,T}^{g}[\bar{z}(B_{T}-B_{t})]$. Then $g(\bar{z})$ is obtained by
(\ref{e2.exm1}).
\end{example}

\begin{remark}
The above test works also for the case $g:[0,\infty)\times\mathbf{R}%
^{d}\rightarrow\mathbf{R}$, or for a more general situation $g=\gamma
y+g_{0}(t,z)$.
\end{remark}

An interesting problem is, in general, how to find the price generating
function $g$ by a testing of the input--output behavior of $\mathbb{E}%
^{g}[\cdot]$? Let $b:R^{n}\longmapsto R^{n}$, $\bar{\sigma}:R^{n}\longmapsto
R^{n\times d}$ be two Lipschitz functions.
\[
X_{s}^{t,x}=x+\int_{t}^{s}b(X_{r}^{t,x})dr+\int_{t}^{s}\sigma(X_{r}%
^{t,x})dB_{r},\; \;s\geq t.
\]
The following result was obtained in Proposition 2.3 of \cite{BCHMP00}.

\begin{proposition}
We assume that the price generating function $g$ satisfies (\ref{h2.1}%
)\textbf{. }We also assume that, for each fixed $(y,z)$, $g(\cdot,y,z)\in
D_{\mathcal{F}}^{2}(0,T)$. Then for each $(t,x,p,y)\in\lbrack0,\infty)\times
R^{n}\times R^{n}\times R$, we have
\[
L^{2}\hbox{--}\lim_{\epsilon\rightarrow0}\frac{1}{\epsilon}[\mathbb{E}%
_{t,t+\epsilon}^{g}[y+p\cdot(X_{t+\epsilon}^{t,x}-x)]-y]=g(t,y,\sigma
^{T}(x)p)+p\cdot b(x).\label{limE}%
\]

\end{proposition}

\section{Pricing an accumulated contingent claim with $\mathbb{E}^{g}%
$--pricing mechanisms\label{ass3}}

\begin{definition}
An accumulated contingent claim $(X,K)\in L^{2}(\mathcal{F}_{T})\times
D_{\mathcal{F}}^{2}(0,T)$ with maturity $T$ is a contract, according which the
writer have to pay the buyer $X$ at $T$ and. in each time interval
$[s,t]\subset\lbrack0,T]$, $K_{t}-K_{s}$.
\end{definition}

\begin{remark}
We understand that, in a real life, $X$ should be non negative and $K$ non
decreasing. But we will see that we can treat the general situation $(X,K)\in
L^{2}(\mathcal{F}_{T})\times D_{\mathcal{F}}^{2}(0,T)$, without any
mathematical obstacle.
\end{remark}

We consider the following BSDE on $[0,t]$ with given terminal condition $X\in
L^{2}(\mathcal{F}_{t})$ and an RCLL process $K\in D_{\mathcal{F}}^{2}%
(0,\infty)$:
\begin{equation}
Y_{s}=X+K_{t}-K_{s}+\int_{s}^{t}g(r,Y_{r},Z_{r})dr-\int_{s}^{t}Z_{r}%
dB_{r},\;s\in\lbrack0,t]. \label{tBSDE}%
\end{equation}
When $K$ is an increasing (resp. decreasing) process, the solution $(Y,Z)$ is
called a $g$--supersolution (resp. $g$--subsolution). These type of solutuions
appear very often in superhedging problem in the pricing of contingent claims
in an incomplete markets, where one need to find the smallest $g$%
--supersolutuion (resp. the largest $g$--subsolution) to replicate $X$. We
first recall the following basic results of BSDE.

\begin{theorem}
\label{th2.1}(\cite{Pardoux-Peng1990}, \cite{Peng1997a}) We assume
(\ref{h2.1}). Then there exists a unique solution $(Y,Z)\in L_{\mathcal{F}%
}^{2}(0,t;R\times R^{d})$ of BSDE (\ref{tBSDE}). We denote it by
\begin{equation}
(Y_{s}^{t,X,K},Z_{s}^{t,X,K})=(Y_{s},Z_{s}),\;s\in\lbrack0,t]. \label{e2.1a}%
\end{equation}
We have
\[
Y^{t,X,K}+K\in S_{\mathcal{F}}^{2}(0,t).
\]

\end{theorem}

\begin{proof}
In \cite{Pardoux-Peng1990} (see also \cite{EPQ1997}), the result of BSDE is
for $t=T$and $K_{t}=\int_{0}^{t}\phi_{s}ds$ for some $\phi\in L_{\mathcal{F}%
}^{2}(0,T)$. The present situation can be treated by setting (see
\cite{Peng1997})
\begin{equation}%
\begin{array}
[c]{rl}%
\bar{Y}_{s} & :=Y_{s}+K_{s},\\
\bar{g}(s,y,z) & :=g(s,y-K_{s},z)1_{[0,t]}(s)
\end{array}
\label{e-gbar}%
\end{equation}
and considering the following equivalent BSDE
\begin{equation}
\bar{Y}_{s}=X+K_{t}+\int_{s}^{T}\bar{g}(r,\bar{Y}_{r},Z_{r})dr-\int_{s}%
^{T}Z_{r}dB_{r},\;s\in\lbrack0,T]. \label{TBSDE}%
\end{equation}
It is clear that $\bar{Y}_{s}\equiv X+K_{s}$, $Z_{s}\equiv0$on $[t,T]$. Since
$\bar{g}$ is a Lipschitz function with the same Lipschitz constant $\mu$and
\[
\bar{g}(\cdot,0,0)=g(\cdot,-K_{\cdot},0)1_{[0,t]}(\cdot)\in L_{\mathcal{F}%
}^{2}(0,T),
\]
thus, by \cite{Pardoux-Peng1990}, \cite{Peng1997a}, the BSDE (\ref{TBSDE}) has
a unique solution $(\bar{Y},Z){\normalsize .}$
\end{proof}

We introduce a new notation.

\begin{definition}
\label{d2.2}We denote, for $s\leq t$,
\begin{align}
\mathbb{E}_{s,t}^{g}[X;K_{\cdot}]  &  :=Y_{s}^{t,X,K}\label{e2.3}\\
\mathbb{E}_{s,t}^{g}[X]  &  :=\mathbb{E}_{s,t}^{g}[X;0]. \label{e2.4}%
\end{align}
This notion generalizes that of $\mathbb{E}_{s,t}^{g}[\cdot]$ in Definition
\ref{EgstX}. Clearly when (\ref{h2.2})--(a) is satisfied, we have
$\mathbb{E}_{s,t}^{g}[0]=\mathbb{E}_{s,t}^{g}[0;0]=0$, $0\leq s\leq t\leq T$.
\end{definition}

\begin{remark}
In fact, for each maturity $T\geq0$ the price process of the accumulated
contingent claim $(X,K)\in L^{2}(\mathcal{F}_{T})\times D_{\mathcal{F}}%
^{2}(0,T)$ produced by $\mathbb{E}^{g}[\cdot]$ is $\mathbb{E}_{s,T}^{g}[X;K]$,
$s\leq T$. We will prove it for a more general price mechanism $\mathbb{E}%
[\cdot]$, see the next subsection.
\end{remark}

\begin{remark}
About the notations $\mathbb{E}^{g}[\cdot]$. This notation was firstly
introduced in \cite{Peng1997} in the case where $g$ satisfies (\ref{h2.2}%
)--(b). In this situation it is easy to check that
\[
\mathbb{E}_{s,t}^{g}[X]\equiv\mathbb{E}_{s,T}^{g}[X],\; \forall0\leq s\leq
t\leq T.
\]
In other words, $\mathbb{E}^{g}$--is a nonlinear expectation, called
$g$--expectation. The general situation, i.e., without (\ref{h2.2}) was
introduced in \cite{Peng1997a} and \cite{Chen-Peng2001}.
\end{remark}

By the above existence and uniqueness theorem, we have for each $0\leq r\leq
s\leq t$ and for each $X\in L^{2}(\mathcal{F}_{t})$ and $K\in D_{\mathcal{F}%
}^{2}(0,T)$,
\begin{equation}
\mathbb{E}_{r,s}^{g}[\mathbb{E}_{s,t}^{g}[X;K_{\cdot}];K_{\cdot}%
]=\mathbb{E}_{r,t}^{g}[X;K_{\cdot}],\; \hbox{a.s.} \label{e2.5}%
\end{equation}
It is also easy to check that, with the notation $g_{-}(t,y,z):=-g(t,-y,-z)$
\begin{equation}
-\mathbb{E}_{s,t}^{g}[X;K_{\cdot}]=\mathbb{E}_{s,t}^{g_{-}}[-X;-K_{\cdot}].
\label{e2.5.1}%
\end{equation}

We will see that $\{ \mathbb{E}_{t,T}^{g}\left[  X\right]  \}_{0\leq t\leq T}%
$, $X\in L^{2}(\mathcal{F}_{T})$ form an $\mathcal{F}_{t}$--consistent
nonlinear pricing mechanism. The following monotonicity property is the
comparison theorem of BSDE.

\begin{theorem}
\label{th2.2}We assume (\ref{h2.1}). For each fixed maturity let for let
$(X,K)$ and $(X^{\prime},K^{\prime})$ $\in L^{2}(\mathcal{F}_{t})\times
D_{\mathcal{F}}^{2}(0,T)$ be two accumulated contingent claims satisfying
$X\geq X^{\prime}$ and that $K-K^{\prime}$ is an increasing process. Then we
have
\begin{equation}
\mathbb{E}_{s,t}^{g}[X;K_{\cdot}]\geq\mathbb{E}_{s,t}^{g}[X^{\prime};K_{\cdot
}^{\prime}],\ \ \forall s\leq t. \label{e2.6}%
\end{equation}
In particular,
\begin{equation}
\mathbb{E}_{s,t}^{g}[X]\geq\mathbb{E}_{s,t}^{g}[X^{\prime}]. \label{e2.7}%
\end{equation}
If $A\in D_{\mathcal{F}}^{2}(0,T)$ is an increasing process, then
\begin{equation}
\mathbb{E}_{s,t}^{g}[X;A_{\cdot}]\geq\mathbb{E}_{s,t}^{g}[X]. \label{e2.8}%
\end{equation}

\end{theorem}

\begin{proof}
\textbf{ }The case $K_{t}\equiv K_{t}^{\prime}\equiv0$ is the classical
comparison theorem of BSDE. The present general situation, see
\cite{Peng1997a} or \cite{Peng2003b}.
\end{proof}

We recall the special price generating function $g_{\mu}(y,z)$ defined in
(\ref{e2.9}). It is a very strong generating function. In fact we have

\begin{corollary}
\label{c2.1}The $g$--pricing mechanism $\mathbb{E}^{g}$ is dominated by
$\mathbb{E}^{g_{\mu}}$ in the following sense: for each $t\geq0$, let $(X,K)$
and $(X^{\prime},K^{\prime})$ $\in L^{2}(\mathcal{F}_{t})\times D_{\mathcal{F}%
}^{2}(0,T)$ be two accumulated contingent claims with maturity $t$, then we
have
\begin{equation}
\mathbb{E}_{s,t}^{g}[X;K_{\cdot}]-\mathbb{E}_{s,t}^{g}[X^{\prime};K_{\cdot
}^{\prime}]\leq\mathbb{E}_{s,t}^{g_{\mu}}[X-X^{\prime};K_{\cdot}-K_{\cdot
}^{\prime}] \label{e2.10}%
\end{equation}
where $\mu$ is the Lipschitz constant of $g$ given in (\ref{h2.1}). In
particular, since $g_{\mu}$ the generating function $g_{\mu}$ itself has
Lipschitz constant $\mu$, we have
\begin{equation}
\mathbb{E}_{s,t}^{g_{\mu}}[X;K_{\cdot}]-\mathbb{E}_{s,t}^{g_{\mu}}[X^{\prime
};K_{\cdot}^{\prime}]\leq\mathbb{E}_{s,t}^{g_{\mu}}[X-X^{\prime};K_{\cdot
}-K_{\cdot}^{\prime}] \label{e2.10aa}%
\end{equation}

\end{corollary}

\begin{proof}
By the definition, The pricing processes produces by $Y_{s}=\mathbb{E}%
_{s,t}^{g}[X;K_{\cdot}]$ and $Y_{s}^{\prime}=\mathbb{E}_{s,t}^{g}[X^{\prime
};K_{\cdot}^{\prime}]$ solve respectively the following BSDEs on $[0,t]$:
\begin{align*}
Y_{s}  &  =X+K_{t}-K_{s}+\int_{s}^{t}g(r,Y_{r},Z_{r})dr-\int_{s}^{t}%
Z_{r}dB_{r},\\
Y_{s}^{\prime}  &  =X^{\prime}+K_{t}^{\prime}-K_{s}^{\prime}+\int_{s}%
^{t}g(r,Y_{r}^{\prime},Z_{r}^{\prime})dr-\int_{s}^{t}Z_{r}^{\prime}dB_{r}.
\end{align*}
We denote $\hat{Y}=Y-Y^{\prime}$, $\hat{Z}=Z-Z^{\prime}$ and
\[
\hat{K}_{t}=K_{t}-K_{t}^{\prime}+\int_{0}^{t}[-g_{\mu}(\hat{Y}_{s},\hat{Z}%
_{s})+g(s,Y_{s},Z_{s})-g(s,Y_{s},Z_{s})]ds.
\]
Then $(\hat{Y},\hat{Z})$ solves a new BSDE
\[
\hat{Y}_{s}=X-X^{\prime}+\hat{K}_{t}-\hat{K}_{s}+\int_{s}^{t}g_{\mu}(r,\hat
{Y}_{r},\hat{Z}_{r})dr-\int_{s}^{t}\hat{Z}_{r}dB_{r}.
\]
We compare it to the BSDE
\[
\bar{Y}_{s}=X-X^{\prime}+(K-K^{\prime})_{t}-(K-K^{\prime})_{s}+\int_{s}%
^{t}g_{\mu}(r,\bar{Y}_{r},\bar{Z}_{r})dr-\int_{s}^{t}\bar{Z}_{r}dB_{r}.
\]
Since $d(K-K^{\prime}-\hat{K})_{s}\geq0$, thus, by comparison theorem, i.e.,
Theorem \ref{th2.2}, $\bar{Y}_{s}\geq\hat{Y}_{s}=Y_{s}-Y_{s}^{\prime}$. We
thus have (\ref{e2.10}).\
\end{proof}

It is very interesting to observe that, given a price generating function $g$,
$\mathbb{E}_{s,t}^{g}[\cdot;K_{\cdot}]$ is again an $\mathcal{F}_{t}%
$--consistent pricing mechanism:

\begin{proposition}
\label{p2.1}Let the generating function $g$ satisfies (\ref{h2.1}) and for a
fixed $K\in D_{\mathcal{F}}^{2}(0,\infty)$,
\begin{equation}
\mathbb{E}_{s,t}^{g}[X;K_{\cdot}]:X\in L^{2}(\mathcal{F}_{t})\rightarrow
L^{2}(\mathcal{F}_{s}),\;0\leq s\leq t<\infty\label{e2.11}%
\end{equation}
defined in (\ref{e2.3}) is an $\mathcal{F}_{t}$--consistent pricing mechanism,
i.e., it satisfies (A1)--(A4) of Definition \ref{d2.1}.
\end{proposition}

\begin{proof}
(A1) is given by (\ref{e2.6}). (A2) is clearly true by the definition. (A3) is
proved by (\ref{e2.5}). We now consider (A4), i.e., for each $t$ and $X\in
L^{2}(\mathcal{F}_{t})$, we have
\begin{equation}
1_{A}\mathbb{E}_{s,t}^{g}[X;K_{\cdot}]=1_{A}\mathbb{E}_{s,t}^{g}%
[1_{A}X;K_{\cdot}],\; \; \forall s\leq t,\ A\in\mathcal{F}_{t}. \label{e2.12}%
\end{equation}
as well as
\begin{equation}
1_{A}\mathbb{E}_{s,t}^{g}[X;K_{\cdot}]=\mathbb{E}_{s,t}^{g_{s,A}}%
[1_{A}X;K_{\cdot}^{s,A}],\; \; \forall A\in\mathcal{F}_{t}, \label{e2.12a}%
\end{equation}
where we set
\begin{align}
g_{s,A}(t,y,z)  &  :=1_{[0,s)}(t)g(t,y,z)+1_{[s,T]}(t)1_{A}%
g(t,y,z),\label{e2.12b}\\
K_{t}^{s,A}  &  :=1_{[0,s)}(t)K_{t}+1_{[s,T]}(t)1_{A}(K_{t}-K_{s}).
\label{e2.12c}%
\end{align}
We will give the proof of (\ref{e2.12}). The proof of (\ref{e2.12a}) is
similar. According to BSDE (\ref{tBSDE}) for each time $r\in\lbrack s,t]$,
$Y_{r}:=\mathbb{E}_{r,s}^{g}[X;K_{\cdot}]$ and $\bar{Y}_{r}:=\mathbb{E}%
_{s,t}^{g}[1_{A}X;K_{\cdot}]$ solve respectively
\[
Y_{r}=X+K_{t}-K_{r}+\int_{r}^{t}g(r,Y_{u},Z_{u})du-\int_{r}^{t}Z_{u}dB_{u},
\]
and
\[
\bar{Y}_{r}=1_{A}X+K_{t}-K_{r}+\int_{r}^{t}g(u,\bar{Y}_{u},\bar{Z}_{u}%
)du-\int_{r}^{t}\bar{Z}_{u}dB_{u}%
\]
We multiply $1_{A}$, $A\in\mathcal{F}_{s}$ on both sides of the above two
BSDEs. Since $1_{A}g(r,Y_{r},Z_{r})=1_{A}g(r,Y_{r}1_{A},Z_{r}1_{A})$, we have
\[
1_{A}Y_{r}=1_{A}X+1_{A}K_{t}-1_{A}K_{r}+\int_{r}^{t}1_{A}g(u1_{A}Y_{u}%
,1_{A}Z_{u})du-\int_{r}^{t}1_{A}Z_{u}dB_{u},
\]
and
\[
1_{A}\bar{Y}_{r}=1_{A}X+1_{A}K_{t}-1_{A}K_{r}+\int_{r}^{t}1_{A}g(u1_{A}%
Y_{u},1_{A}Z_{u})du-\int_{r}^{t}1_{A}\bar{Z}_{u}dB_{u}.
\]
It is clear that $1_{A}Y_{r}$ and $1_{A}\bar{Y}_{r}$ satisfy exactly the same
BSDE with the same terminal condition on $[s,t]$. By uniqueness of BSDE,
$1_{A}Y_{r}\equiv1_{A}\bar{Y}_{r}$ on $[s,t]$, i.e., $1_{A}\mathbb{E}%
_{s,t}^{g}[X;K_{\cdot}]\equiv1_{A}\mathbb{E}_{s,t}^{g}[1_{A}X;K_{\cdot}]$. The
proof is complete.
\end{proof}

If $Y$ is the data of a price process produced by some contingent claim
pricing mechanism, in many situations it is practically meaningful and
financially interesting to compare this data by using a given $g$--pricing
mechanism. One typical situation is that the price produced by $\mathbb{E}%
^{g}$ is weaker (resp. stronger). In this situation $Y$ is called a
$g$--supermartingale (resp. $g$--submartingale). Here the term
\textquotedblleft$g$--martingale\textquotedblright\ is a nonlinear, and
nontrivial generalization of the classical one, due to the similarity between
the classical conditional expectation $\mathbf{E}[\cdot|\mathcal{F}_{s}]$ and
$\mathbb{E}_{s,t}^{g}[\cdot]$: .

\begin{definition}
\label{d2.3}Let $K\in D_{\mathcal{F}}^{2}(0,\infty)$ be given. A process $Y\in
D_{\mathcal{F}}^{2}(0,\infty)$ is said to be an $\mathbb{E}^{g}[\cdot
;K]$--martingale (resp. $\mathbb{E}^{g}[\cdot;K]$--supermartingale,
$\mathbb{E}^{g}[\cdot;K]$--submartingale) if for each $0\leq s\leq t$%
\begin{equation}
\mathbb{E}_{s,t}^{g}[Y_{t};K_{\cdot}]=Y_{s}\hbox{, (resp. }\leq Y_{s}%
\hbox{, }\geq Y_{s}\hbox{)}. \label{e2.13}%
\end{equation}

\end{definition}

Clearly a $\mathbb{E}^{g}$--martingale $Y$ is a price process produced by this
pricing mechanism: $Y_{s}=\mathbb{E}_{s,T}^{g}[Y_{T}]$, $s\leq T$.

\begin{remark}
\label{r2.1}If $(y,z)\in L_{\mathcal{F}}^{2}(0,T;R\times R^{d})$ solves the
BSDE
\[
y_{s}=y_{t}+K_{t}-K_{s}+\int_{s}^{t}g(r,y_{r},z_{r})dr-\int_{s}^{t}z_{r}%
dB_{r},\hbox{
}s\leq t.
\]
It is clear that $(-y,-z)$ solves
\begin{align*}
-y_{s}  &  =-y_{t}+(-K_{t})-(-K_{s})\\
&  \ \ \ \ \ \ \ \ +\int_{s}^{t}[-g(r,-(-y_{r}),-(-z_{r}))dr-\int_{s}%
^{t}(-z_{r})dB_{r}.
\end{align*}
Thus, if $y$ is an $\mathbb{E}^{g}[\cdot;K_{\cdot}]$--martingale (resp.
$\mathbb{E}^{g}[\cdot;K]$--supermartingale, $\mathbb{E}^{g}[\cdot
;K]$--submartingale), then $-y$ is an $\mathbb{E}_{s,t}^{g_{\ast}}%
[\cdot;-K_{\cdot}]$--martingale (resp. $\mathbb{E}^{g_{\ast}}[\cdot
;K]$--submartingale, $\mathbb{E}^{g_{\ast}}[\cdot;K]$--supermartingale), where
we denote
\[
g_{\ast}(t,y,z):=-g(t,-y,-z).
\]
Therefor many results concerning $\mathbb{E}^{g}[\cdot;K]$--supermartingales
can be also applied to situations of submartingales.
\end{remark}

\begin{example}
Let $X\in L^{2}(\mathcal{F}_{T})$ and $A\in D_{\mathcal{F}}^{2}(0,T)$ be given
such that $A$ is an increasing process. By the monotonicity of $\mathbb{E}%
^{g}$, i.e., Theorem \ref{th2.2}, we have, for $t\in\lbrack0,T]$,
\[%
\begin{array}
[c]{l}%
Y_{t}:=\mathbb{E}_{t,T}^{g}[X]=\mathbb{E}_{t,T}^{g}%
[X;0]\hbox{\ is a }\mathbb{E}^{g}\hbox{--martingale,}\\
Y_{t}^{+}:=\mathbb{E}_{t,T}^{g}[X;A]\hbox{\ is a }\mathbb{E}^{g}%
\hbox{--supermartingale,}\\
Y_{t}^{-}:=\mathbb{E}_{t,T}^{g}[X;-A]\hbox{ is a }\mathbb{E}^{g}%
\hbox{--submartingale.}
\end{array}
\]

\end{example}

As in classical situations, an interesting and hard problem is the inverse
one: if $Y$ is an $\mathbb{E}^{g}$--supermartingale, can we find an increasing
and predictable process $A$ such that $Y_{t}\equiv\mathbb{E}_{t,T}^{g}[X;A]$?
This nonlinear version of Doob--Meyer's decomposition theorem will be stated
as follows. It plays a crucially important role in this paper.

We have the following $\mathbb{E}^{g}$--supermartingale decomposition theorem
of Doob--Meyer's type. This nonlinear decomposition theorem was obtained in
\cite{Peng1999}. But the formulation using the notation $\mathbb{E}_{t,T}%
^{g}[\cdot;A]$ is new. In fact we think this is the intrinsic formulation
since it becomes necessary in the more abstract situation of the $\mathbb{E}%
$--supermartingale decomposition theorem, i.e., Theorem \ref{m6.1} which can
considered as a generalization of the following result.

\begin{proposition}
\label{p2.3} We assume (\ref{h2.1})--(i) and (ii). Let $Y\in D_{\mathcal{F}%
}^{2}(0,T)$ be an $\mathbb{E}^{g}$--supermartingale. Then there exists a
unique increasing process $A\in D_{\mathcal{F}}^{2}(0,T)$ (thus predictable)
with $A_{0}=0$, such that
\begin{equation}
Y_{t}=\mathbb{E}_{t,T}^{g}[Y_{T};A],\; \forall0\leq t\leq T. \label{e2.15}%
\end{equation}

\end{proposition}

\begin{corollary}
\label{c2.2}Let $K\in D_{\mathcal{F}}^{2}(0,T)$ be given and let $Y\in
D_{\mathcal{F}}^{2}(0,T)$ be an $\mathbb{E}^{g}[\cdot;K]$--supermartingale in
the following sense
\begin{equation}
\mathbb{E}_{s,t}^{g}[Y_{t};K]\leq Y_{s},\; \forall0\leq s\leq t\leq T.
\label{e2.16}%
\end{equation}
Then there exists a unique increasing process $A\in D_{\mathcal{F}}^{2}(0,T)$
with $A_{0}=0$, such that
\begin{equation}
Y_{t}=\mathbb{E}_{t,T}^{g}[Y_{T};K+A],\; \forall0\leq t\leq T. \label{e2.17}%
\end{equation}

\end{corollary}

\begin{proof}
By the notations of (\ref{e-gbar}) with $t=T$, we have
\begin{equation}
\mathbb{E}_{s,t}^{g}[Y_{t};K]+K_{s}=\mathbb{E}_{s,t}^{\bar{g}}[Y_{t}+K_{t}].
\label{e2.18}%
\end{equation}
It follows that (\ref{e2.16}) is equivalent to
\begin{equation}
\mathbb{E}_{s,t}^{\bar{g}}[Y_{t}+K_{t}]\leq Y_{s}+K_{s},\; \forall0\leq s\leq
t\leq T. \label{e2.19}%
\end{equation}
In other words, $Y+K$ is an $\mathbb{E}^{\bar{g}}$--supermartingale in the
sense of (\ref{e2.13}). By the above supermartingale decomposition theorem,
Proposition \ref{p2.3}, there exists an increasing process $A\in
D_{\mathcal{F}}^{2}(0,T)$ with $A_{0}=0$, such that
\begin{equation}
Y_{t}+K_{t}=\mathbb{E}_{t,T}^{\bar{g}}[Y_{T}+K_{T};A],\; \forall0\leq t\leq T,
\label{e2.20}%
\end{equation}
or, equivalently (\ref{e2.17}).
\end{proof}

\section{Characterization of $g$-pricing mechanism by its generating function
$g$}

For a price mechanism $\mathbb{E}^{g}[\cdot]$, it is important to distinct her
selling price and buying price. We now fix that $\mathbb{E}^{g}[\cdot]$ is the
selling price. A rational price mechanism must be%
\[
\mathbb{E}^{g}[X]\geq-\mathbb{E}^{g}[-X].
\]
It also possesses some other properties, such as convexity, or moreover,
sub-additivity. See \cite{ADEH1999}, \cite{El-Bar}, \cite{El-Bar2005},
\cite{Chen-Epstein2002}, \cite{EQ1995}, \cite{EPQ1997}, \cite{Fo-Sc},
\cite{Fritteli}, \cite{Peng-Yang}, \cite{rosazza}, etc. for the ecomomic
meanings. An interesting question is: what the corresponding generating
function $g$ will behaves if the the $\mathbb{E}^{g}$ satisfies the above
properties. We will see that $g$ perfectly reflects the behavior of
$\mathbb{E}^{g}$. This will be very important for using data of the pricing
processes to statistically find $g$. We begin with introducing some technique lemmas.

Let a functions $f:(\omega,t,y,z)\in\Omega\times\lbrack0,T]\times R\times
R^{d}\rightarrow R$ satisfy the same Lipschitz condition (\ref{h2.1}) as for
$g$. For each $n=1,2,3,\cdots$, we set
\begin{align}
f^{n}(s,y,z)  &  :=\sum_{i=0}^{2^{n}-1}f(s,Y_{s}^{t_{i}^{n},y},z)1_{[t_{i}%
^{n},t_{i+1}^{n})}(s),\;s\in\lbrack0,T]\label{efnyz}\\
t_{i}^{n}  &  =i2^{-n}T,i=0,1,2,\cdots,2^{n}.
\end{align}
It is clear that $f^{n}$ is an $\mathcal{F}_{t}$--adapted process.

For each fixed $(t,y,z)\in\lbrack0,T]\times R\times R^{d}$, we consider the
following SDE of It\^{o}'s type defined on $[t,T]$:
\begin{equation}
Y_{s}^{t,y,z}=y-\int_{t}^{s}f(r,Y_{r}^{t,y,z},z)dr+z(B_{s}-B_{t}) \label{Ytyz}%
\end{equation}
We have the following classical result of It\^{o}'s SDE.

\begin{lemma}
\label{em7.2a}We assume that $f$ satisfies the same Lipschitz condition
(\ref{h2.1}) as for $g$. (\ref{h2.1}). Then there exists a constant $C$,
depending only on $\mu$, $T$ and $E\int_{0}^{T}|f(\cdot,0,0)|^{2}ds$, such
that, for each $(t,y,z)\in\lbrack0,T]\times R\times R^{d}$, we have
\begin{equation}
E[|Y_{s}^{t,y,z}-y|^{2}]\leq C(|y|^{2}+|z|^{2}+1)(s-t),\; \forall s\in\lbrack
t,T]. \label{ee7.15}%
\end{equation}

\end{lemma}

\begin{proof}
\textbf{ }It is classic that $E\int_{0}^{T}|f(r,Y_{r}^{t,y,z},z)|dr^{2}\leq
C_{0}(|y|^{2}+|z|^{2}+1)$, where $C_{0}$ depends only on $\mu$, $T$ and
$E\int_{0}^{T}|f(\cdot,0,0)|^{2}ds$. We then have
\begin{align*}
E[|Y_{s}^{t,y,z}-y|^{2}]  &  \leq2E[|\int_{t}^{s}f(r,Y_{r}^{t,y,z}%
,z)dr|^{2}]+2|z|^{2}(s-t)\\
\  &  \leq2E[|\int_{t}^{s}f(r,Y_{r}^{t,y,z},z)dr|^{2}]+2|z|^{2}(s-t)\\
\  &  \leq2E[\int_{t}^{s}|f(r,Y_{r}^{t,y,z},z)|^{2}dr](t-s)+2|z|^{2}(s-t)\\
\  &  \leq C(|y|^{2}+|z|^{2}+1)(s-t).
\end{align*}

\end{proof}

\begin{lemma}
\label{FntoF}For each fixed $(y,z)\in R\times R^{d}$, $\{f^{n}(\cdot
,y,z)\}_{n=1}^{\infty}$ converges to $f(\cdot,y,z)$ in $L_{\mathcal{F}}%
^{2}(0,T)$, i.e.,
\begin{equation}
\lim_{n\rightarrow\infty}E\int_{0}^{T}|f^{n}(s,y,z)-f(s,y,z)|^{2}ds=0.
\label{fntof}%
\end{equation}

\end{lemma}

\begin{proof}
For each $s\in\lbrack0,T)$, there are some integers $i\leq2^{n}-1$ such that
$s\in\lbrack t_{i}^{n},t_{i+1}^{n})$. We have, by (\ref{ee7.15})
\begin{align*}
E[|f^{n}(s,y,z)-f(s,y,z)|^{2}]  &  =E[|f(s,Y_{s}^{t_{i}^{n},y}%
,z)-f(s,y,z)|^{2}]\\
\  &  \leq\mu^{2}E[|Y_{s}^{t_{i}^{n},y,z}-y|^{2}]\\
\  &  \leq\mu^{2}C(|y|^{2}+|z|^{2}+1)2^{-n}T.
\end{align*}
Thus $\{f^{n}(\cdot,y,z)\}_{n=1}^{\infty}$ converges to $f(\cdot,y,z)$ in
$L_{\mathcal{F}}^{2}(0,T)$.
\end{proof}

\begin{lemma}
\label{Lem-f-f-}Let $f:(\omega,t,y,z)\in\Omega\times\lbrack0,T]\times R\times
R^{d}\rightarrow R$ satisfies the same Lipschitz condition (\ref{h2.1}) as for
$g$. If for each $(t,y,z)\in\lbrack0,T]\times R\times R^{d}$, we have
\[
f(\omega,r,Y_{r}^{t,y,z},z)\geq0\ \ \text{(resp. }=0\text{)},\;(\omega
,r)\in\lbrack t,T]\times\Omega\text{, .}dr\times dP\text{-a.s..}%
\]
Then We then, for each $(y,z)\in R\times R^{d}$
\begin{equation}
f(\omega,t,y,z)\geq0,\; \text{(resp. }=0\text{)},\ (\omega,t)\in
\lbrack0,T]\times\Omega\text{, .}dt\times dP\text{-a.s..} \label{f-f-}%
\end{equation}

\end{lemma}

\begin{proof}
Let us fix $y$ and $z$. We define $f^{n}(s,y,z)$ as in (\ref{efnyz}). It is
clear that,
\[
f^{n}(r,y,z)\geq0,\ \text{(resp. }=0\text{),}\;(\omega,r)\in\lbrack
0,T]\times\Omega\text{, .}dr\times dP\text{ a.s.}%
\]
But from Lemma \ref{FntoF} we have $f^{n}(\cdot,y,z)\rightarrow f(\cdot,y,z)$,
in $L_{\mathcal{F}}^{2}(0,T)$ as $n\rightarrow\infty$. We thus have
\textbf{\ref{f-f-}}.
\end{proof}

We need the following inverse comparison theorem which generalizes the results
of \cite{BCHMP00} and \cite{CHMP2001} in the sense that $g$ does not need to
be continuous, or right continuous, in time. We thus finally obtain an
equivalent conditions under the standard condition (\ref{h2.1}) of BSDE. We
notice that this result was obtained by already by \cite{JiangL2004a} and
\cite{JiangL2004b}. Here we will use a very different method that will be
applied in the proof of our main theorem.

\begin{proposition}
\label{invcompar}Let $g$, $\bar{g}:$ $(\omega,t,y,z)\in\Omega\times
\lbrack0,T]\times R\times R^{d}\rightarrow R$ be two price generating
functions satisfying (\ref{h2.1}). Then the following two conditions are
equivalent: \newline\textbf{(i) }$g(\omega,t,y,z)\geq\bar{g}(\omega,t,y,z)$,
$\forall(y,z)\in R\times R^{d}$, $dP\times dt\ $a.s. \newline\textbf{(ii) }The
corresponding pricing mechanisms $\mathbb{E}^{g}[\cdot]$, $\mathbb{E}^{\bar
{g}}[\cdot]$ satisfy
\[
\mathbb{E}_{s,t}^{g}[\xi]\geq\mathbb{E}_{s,t}^{\bar{g}}[\xi],\ \forall0\leq
s\leq t,\ \forall\xi\in L^{2}(\mathcal{F}_{t})\
\]

\end{proposition}

\begin{proof}
The method of the proof is significantly different from \cite{BCHMP00} and
\cite{CHMP2001}. \newline The part \textbf{(i)}$\Rightarrow$\textbf{(ii) }is
simply from the standard comparison theorem of BSDE. We now prove the part
\textbf{(ii)}$\Rightarrow$\textbf{(i). }For each fixed $(t,y,z)\in
\lbrack0,T]\times R\times R^{d}$, the solution $(Y_{s}^{t,y,z})_{s\in\lbrack
t,T]}$ of SDE with $f=g$ is a $\bar{g}$--supermartingale. By the decomposition
theorem, there exists an increasing process $A\in D_{\mathcal{F}}^{2}(0,T)$
such that
\[
Y_{s}^{t,y,z}=y-\int_{t}^{s}\bar{g}(r,Y_{r}^{t,y,z},z)dr-(A_{s}-A_{t}%
)+z(B_{s}-B_{t}),\;
\]
Comparing this with $Y_{s}^{t,y,z}=y-\int_{t}^{s}g(r,Y_{r}^{t,y,z}%
,z)dr+z(B_{s}-B_{t})$, we have
\[
g(r,Y_{r}^{t,y,z},z)\geq\bar{g}(r,Y_{r}^{t,y,z},z),\; \text{a.e, in
}[t,T]\text{, a.s..}%
\]
We then can apply Lemma \ref{Lem-f-f-} to obtain (i).
\end{proof}

\begin{corollary}
The following two conditions are equivalent: \newline\textbf{(i)} The price
generating function $g$ satisfies. for each $(y,z)\in R\times R^{d}$,
\[
g(t,y,z)\geq-g(t,-y,-z),\; \text{a.e., a,s,,}%
\]
\textbf{(ii)} $\mathbb{E}_{s,t}^{g}[\cdot]:L^{2}(\mathcal{F}_{t})\longmapsto
L^{2}(\mathcal{F}_{s})$ is a seller's pricing mechanism, i.e., for each $0\leq
s\leq t$ , $\mathbb{E}_{s,t}^{g}[\xi]\geq-\mathbb{E}_{s,t}^{g}[-\xi]$, for
each $\xi\in L^{2}(\mathcal{F}_{t})$.
\end{corollary}

\begin{proof}
We denote $\bar{g}(t,y,z):=-g(t,-y,-z)$ and compare the following two BSDE:%
\[
Y_{s}=\xi+\int_{s}^{t}g(r,Y_{r},Z_{r}dr-\int_{s}^{t}Z_{r}dB_{r},\ s\in
\lbrack0,t],
\]
and
\[
\bar{Y}_{s}=\xi+\int_{s}^{t}\bar{g}(r,\bar{Y},\bar{Z}_{r})dr-\int_{s}^{t}%
\bar{Z}_{r}dB_{r},\ s\in\lbrack0,t].
\]
By the above Proposition, one has $\mathbb{E}_{s,t}^{g}[\cdot]\geq
\mathbb{E}_{s,t}^{\bar{g}}[\cdot]$, iff $g\geq\bar{g}$. This with
$\mathbb{E}_{s,t}^{\bar{g}}[\xi]=-\mathbb{E}_{s,t}^{g}[-\xi]$ yields (i)
$\Leftrightarrow$ (ii).
\end{proof}

\begin{proposition}
The following two conditions are equivalent: \newline\textbf{(i)} The price
generating function $g=g(t,y,z)$ is convex (resp. concave) in $(y,z)$, i.e.,
for each $(y,z)$ and $(y^{\prime},z^{\prime})$ in $R\times R^{d}$ and for a.e.
$t\in\lbrack0,T]$
\begin{align*}
g(s,\alpha y+(1-\alpha)y^{\prime},\alpha z+(1-\alpha)z^{\prime})  &
\leq\alpha g(s,y,z)+(1-\alpha)g(s,y^{\prime},z^{\prime}),\; \text{a.s.}\\
\; \; \text{(resp.}  &  \geq\alpha g(s,y,z)+(1-\alpha)g(s,y^{\prime}%
,z^{\prime}),\; \text{a.s.).}%
\end{align*}
\textbf{(ii)} The corresponding pricing mechanism $\mathbb{E}_{s,t}^{g}%
[\cdot]$ is a convex (resp. concave), i.e., for each fixed $\alpha\in
\lbrack0,1]$, we have
\begin{align}
\mathbb{E}_{s,t}^{g}[\alpha\xi+(1-\alpha)\zeta]  &  \leq\alpha\mathbb{E}%
_{s,t}^{g}[\xi]+(1-\alpha)\mathbb{E}_{s,t}^{g}[\zeta],\; \text{a.s.}%
\label{Econvex}\\
\text{(resp. }  &  \geq\alpha\mathbb{E}_{s,t}^{g}[\xi]+(1-\alpha
)\mathbb{E}_{s,t}^{g}[\zeta],\; \text{a.s.)}\nonumber\\
\forall s  &  \leq t,\ \forall\xi,\zeta\in L^{2}(\mathcal{F}_{t}).\nonumber
\end{align}

\end{proposition}

\begin{proof}
\textbf{ }We only prove the convex case. \newline\textbf{(i)}$\Rightarrow
$\textbf{(ii): }For a given $t>0$, we set $Y_{s}^{\xi}:=\mathbb{E}_{s,t}%
^{g}[\xi]$, $Y_{s}^{\zeta}:=\mathbb{E}_{s,t}^{g}[\zeta]$, $s\in\lbrack0,t]$.
These two pricing processes solve respectively the following two BSDEs on
$[0,t]$:
\begin{align*}
Y_{s}^{\xi}  &  =\xi+\int_{s}^{t}g(r,Y_{r}^{\xi},Z_{r}^{\xi})dr-\int_{s}%
^{t}Z_{r}^{\xi}dB_{r},\\
Y_{s}^{\zeta}  &  =\zeta+\int_{s}^{t}g(r,Y_{r}^{\zeta},Z_{r}^{\zeta}%
)dr-\int_{s}^{t}Z_{r}^{\zeta}dB_{r}.
\end{align*}
Their convex combination: $(Y_{s},Z_{s}):=(\alpha Y_{s}^{\xi}+(1-\alpha
)Y_{s}^{\zeta},\alpha Z_{s}^{\xi}+(1-\alpha)Z_{s}^{\zeta})$, satisfies
\begin{align*}
Y_{s}  &  =\alpha\xi+(1-\alpha)\zeta+\int_{s}^{t}[g(r,Y_{r},Z_{r})+\psi
_{r}]dr-\int_{s}^{t}Z_{r}dB_{r},\\
\text{where we set }\psi_{r}  &  =\alpha g(r,Y_{r}^{\xi},Z_{r}^{\xi
})+(1-\alpha)g(r,Y_{r}^{\zeta},Z_{r}^{\zeta})-g(r,Y_{r},Z_{r})\text{.}%
\end{align*}
But since the price generating function $g$ is convex in $(y,z)$, we have
$\psi\geq0$. It then follows from the comparison theorem that $Y_{s}%
\geq\mathbb{E}_{s,t}^{g}[\alpha\xi+(1-\alpha)\zeta]$, for each. We thus have
\textbf{(ii)}. \newline\textbf{(ii)}$\Rightarrow$\textbf{(i): }Let $Y^{t,y,z}$
be the solution of SDE (\ref{Ytyz}). For fixed $t\in\lbrack0,T)$ and $(y,z)$,
$(y^{\prime},z^{\prime})$ in $R\times R^{d}$, we have
\[
Y_{s}^{t,y,z}=\mathbb{E}_{s,t}^{g}[Y_{t}^{t,y,z}],\;Y_{s}^{t,y^{\prime
},z^{\prime}}=\mathbb{E}_{s,t}^{g}[Y_{t}^{t,y^{\prime},z^{\prime}}].
\]
We set $Y_{s}:=\alpha Y_{s}^{t,y,z}+(1-\alpha)Y_{s}^{t,y^{\prime},z^{\prime}}%
$, $s\in\lbrack t_{0},T]$. By (\ref{Econvex}),
\begin{align*}
\mathbb{E}_{s,t}^{g}[Y_{t}]  &  \leq\alpha\mathbb{E}_{s,t}^{g}[Y_{s}%
^{t,y,z}]+(1-\alpha)\mathbb{E}_{s,t}^{g}[Y_{s}^{t,y^{\prime},z^{\prime}}]\\
&  =\alpha Y_{s}^{t,y,z}+(1-\alpha)Y_{s}^{t,y^{\prime},z^{\prime}}\\
&  =Y_{s}.
\end{align*}
Thus the process $Y$ is a $g$--supermartingale defined on $[t,T]$. It follows
from the decomposition theorem, i.e., Theorem \ref{p2.3}, that, there exists
an increasing process $A\in D_{\mathcal{F}}^{2}(t,T)$ such that
\[
Y_{s}=Y_{t}-\int_{t}^{s}g(r,Y_{r},Z_{r})dr-(A_{s}-A_{t})+\int_{t}^{s}%
Z_{r}dB_{s}.
\]
We compare this with
\begin{align*}
Y_{s}  &  =\alpha Y_{s}^{t,y,z}+(1-\alpha)Y_{s}^{t,y^{\prime},z^{\prime}}\\
&  =\alpha y+(1-\alpha)y^{\prime}-\int_{t}^{s}[\alpha g(r,Y_{r}^{t,y,z}%
,z)+(1-\alpha)g(r,Y_{r}^{t,y^{\prime},z^{\prime}},z^{\prime})]dr\\
&  +(\alpha z+(1-\alpha)z^{\prime})(B_{s}-B_{t}),
\end{align*}
It follows that
\begin{align*}
Y_{t}  &  =\alpha y+(1-\alpha)y^{\prime},\;Z_{r}\equiv\alpha z+(1-\alpha
)z^{\prime},\;\\
g(r,Y_{r},Z_{r})  &  \equiv g(r,\alpha Y_{r}^{t,y,z}+(1-\alpha)Y_{r}%
^{t,y^{\prime},z^{\prime}},\alpha z+(1-\alpha)z^{\prime}).
\end{align*}
Thus we have
\[
g(s,\alpha Y_{s}^{t,y,z}+(1-\alpha)Y_{s}^{t,y^{\prime},z^{\prime}},\alpha
z+(1-\alpha)z^{\prime})\leq\alpha g(s,Y_{s}^{t,y,z},z)+(1-\alpha
)g(s,Y_{s}^{t,y^{\prime},z^{\prime}},z^{\prime}).
\]
We then can apply Lemma \ref{Lem-f-f-} to obtain (i).
\end{proof}

\begin{proposition}
The following two conditions are equivalent: \newline\textbf{(i)} The price
generating function $g$ is positively homogenous in $(y,z)\in R\times R^{d}$,
i.e.,
\[
g(t,\lambda y,\lambda z)=\lambda g(t,y,z),\; \text{a.e., a,s,,}%
\]
\textbf{(ii)} The corresponding pricing mechanism $\mathbb{E}_{s,t}^{g}%
[\cdot]:L^{2}(\mathcal{F}_{t})\longmapsto L^{2}(\mathcal{F}_{s})$ is
positively homogenous: for each $0\leq s\leq t$ , i.e., $\mathbb{E}_{s,t}%
^{g}[\lambda\xi]=\lambda\mathbb{E}_{s,t}^{g}[\xi]$, for each $\lambda\geq0$
and $\xi\in L^{2}(\mathcal{F}_{t})$.
\end{proposition}

\begin{proof}
(i)$\Rightarrow$(ii) is easy.\newline(ii)$\Rightarrow$(i):\textbf{ }Let
$Y^{t,y,z}$ be the solution of SDE (\ref{Ytyz}). For fixed $t\in\lbrack0,T)$
and $(y,z)$ in $R\times R^{d}$, we have $\lambda Y_{s}^{t,y,z}=\mathbb{E}%
_{s,t}^{g}[\lambda Y_{t}^{t,y,z}]$, $s\in\lbrack t,T]$. This implies that,
there exists $Z_{\cdot}^{t,y,z,\lambda}\in L_{\mathcal{F}}^{2}(t,T;R^{d})$,
such that
\[
\lambda Y_{s}^{t,y,z}=\lambda y-\int_{t}^{s}g(r,\lambda Y_{r}^{t,y,z}%
,Z_{r}^{t,y,z,\lambda})dr+\int_{t}^{s}Z_{r}^{t,y,z,\lambda}dB_{r}%
,\ s\in\lbrack t,T].
\]
Compare this with $\lambda Y_{s}^{t,y,z}=\lambda y-\int_{t}^{s}\lambda
g(r,Y_{r}^{t,y,z},z)dr+\int_{t}^{s}\lambda zdr$, it follows that $Z_{\cdot
}^{t,y,z,\lambda}\equiv\lambda z$ and $\lambda g(r,Y_{r}^{t,y,z},z)\equiv
g(r,\lambda Y_{r}^{t,y,z},Z_{r}^{t,y,z,\lambda})$, $r\in\lbrack t,T]$. We then
can apply Lemma \ref{Lem-f-f-} to obtain (i).
\end{proof}

From the above two propositions we immediatly have

\begin{corollary}
The following two conditions are equivalent: \newline\textbf{(i)} The price
generating function $g$ is subadditive: for each $(y,z),\ (y^{\prime
},z^{\prime})\in R\times R^{d}$,
\[
g(\omega,t,y+y^{\prime},z+z^{\prime})\leq g(\omega,t,y,z)+g(\omega
,t,y^{\prime},z^{\prime}),\;dt\times dP\text{, a.s.,}%
\]
\textbf{(ii)} The corresponding pricing mechanism $\mathbb{E}_{s,t}^{g}%
[\cdot]:L^{2}(\mathcal{F}_{t})\longmapsto L^{2}(\mathcal{F}_{s})$ is is
subadditive: for each $0\leq s\leq t$ ,%
\[
\mathbb{E}_{s,t}^{g}[\xi+\xi^{\prime}]\leq\mathbb{E}_{s,t}^{g}[\xi
]+\mathbb{E}_{s,t}^{g}[\xi^{\prime}],\ \ \forall\xi,\xi^{\prime}\in
L^{2}(\mathcal{F}_{t}).
\]

\end{corollary}

\begin{proposition}
The price generating function $g$ is independent of $y$ if and only if, the
corresponding $g$--pricing mechanism is cash invariant, namely, for each
$s\leq t$
\[
\mathbb{E}_{s,t}^{g}[\xi+\eta]=\mathbb{E}_{s,t}^{g}[\xi]+\eta,\;
\text{a.s.},\forall\xi\in L^{2}(\mathcal{F}_{t}),\ \eta\in L^{2}%
(\mathcal{F}_{s}).
\]

\end{proposition}

\begin{proof}
\textbf{ }We first prove the \textquotedblleft If\textquotedblright\ part. For
each fixed $(y,z)\in R\times R^{d}$, we have $Y_{s}^{t,y,z}\equiv
\mathbb{E}_{s,T}^{g}[Y_{T}^{t,y,z}]\equiv y+\mathbb{E}_{s,T}^{g}[Y_{T}%
^{t,y,z}-y]$. Let $\bar{Y}_{s}=\mathbb{E}_{s,T}^{g}[Y_{T}^{t,y,z}-y]$,
$s\in\lbrack0,T]$ and $\bar{Z}$ be the corresponding part of It\^{o}'s
integrand. By $\bar{Y}_{r}\equiv y+Y_{r}^{t,y,z}$ it follows that
\begin{align*}
y+Y_{s}  &  =y+Y_{T}^{t,y,z}+\int_{s}^{T}g(r,Y_{r}^{t,y,z},z)-\int_{s}%
^{T}zdB_{r}\\
&  =(y+Y_{T}^{t,y,z})+\int_{s}^{T}g(r,\bar{Y}_{r},\bar{Z}_{r})-\int_{s}%
^{T}\bar{Z}_{r}dB_{r}.
\end{align*}
Thus $\bar{Z}_{r}\equiv z$ and
\[
g(r,Y_{r}^{t,y,z},z)\equiv g(r,Y_{r}^{t,y,z}-y,\bar{Z}_{r})\equiv
g(r,Y_{r}^{t,y,z}-y,z).
\]
We then can apply Lemma \ref{Lem-f-f-} to obtain that, for each $(y,z)\in
R\times R^{d}$,
\[
g(r,y,z)\equiv g(r,y-y,z)\equiv g(r,0,z).
\]
Namely, $g$ is independent of $y$. \newline\textquotedblleft Only if
part\textquotedblright: For each for each $s\leq t$ and $\xi\in L^{2}%
(\mathcal{F}_{t})$, $\eta\in L^{2}(\mathcal{F}_{s})$, we have
\[
Y_{r}:=\mathbb{E}_{s,t}^{g}[\xi+\eta]=\xi+\eta+\int_{r}^{t}g(u,Z_{u}%
)du-\int_{s}^{t}Z_{u}dB_{u},\;r\in\lbrack s,t].
\]
Thus $\bar{Y}_{r}:=Y_{r}-\eta$ is a $g$--solution on $[s,t]$ with terminal
condition $\bar{Y}_{t}=\xi+\eta$. This implies
\[
\mathbb{E}_{s,t}^{g}[\xi]+\eta=\bar{Y}_{s}=\mathbb{E}_{s,t}^{g}[\xi+\eta].
\]
The proof is complete.
\end{proof}

We consider the following self--financing condition:
\begin{equation}
\mathbb{E}_{s,t}^{g}[0]\equiv0,\; \forall0\leq s\leq t. \label{selffinance}%
\end{equation}

\begin{proposition}
$\mathbb{E}^{g}[\cdot]$ satisfies the self--financing condition if and only if
its price generating function $g$ satisfies
\[
g(t,0,0)=0\text{, a.e., a.s.}%
\]

\end{proposition}

\begin{proof}
The \textquotedblleft if\textquotedblright\ part is obvious. \newline The
\textquotedblleft only if part\textquotedblright: $Y_{t}:=\mathbb{E}_{t,T}%
^{g}[0]\equiv0$, implies
\[
Y_{t}\equiv0\equiv0+\int_{t}^{T}g(s,0,Z_{s})ds-\int_{t}^{T}Z_{s}dB_{s}%
,\;t\in\lbrack0,T].
\]
Thus $Z_{s}\equiv0$ and then $g(s,0,Z_{s})=g(s,0,0)\equiv0$.
\end{proof}

Zero--interesting rate condition:
\[
\mathbb{E}_{s,t}^{g}[\eta]=\eta,\; \forall0\leq s\leq t<\infty\text{, }\eta\in
L^{2}(\mathcal{F}_{s}).
\]

\begin{proposition}
$\mathbb{E}^{g}[\cdot]$ satisfies the zero--interesting rate condition if and
only if its price generating function $g$ satisfies, for each $y\in R$,
\[
g(\cdot,y,0)=0\text{.}%
\]

\end{proposition}

\begin{proof}
\textbf{ }For a fixed $y\in R$, we consider $Y_{s}:=\mathbb{E}_{s,T}%
^{g}[y]\equiv y$. Let $Z_{s}$ be the corresponding It\^{o}'s integrand
\[
Y_{t}=y+\int_{t}^{T}g(s,Y_{s},Z_{s})-\int_{t}^{T}Z_{s}dB_{s}\equiv y.
\]
But this is equivalent to
\[
Y_{t}\equiv y,\;Z_{s}\equiv0\text{, }g(s,y,0)\equiv0.
\]

\end{proof}

For each $\bar z_{\cdot}^{i_{0}}\in L_{\mathcal{F}}^{2}(0,T)$%
\begin{equation}
\mathbb{E}_{t,T}[\xi]+\int_{0}^{t}\bar z_{s}^{i_{0}}dB_{s}^{i_{0}}%
=\mathbb{E}_{t,T}[\xi+\int_{t}^{T}\bar z_{s}^{i_{0}}dB_{s}^{i_{0}}]
\label{zBi0}%
\end{equation}

\begin{proposition}
Condition \ref{zBi0} holds if and only if $g(s,y,z)$ does not depends on the
$i_{0}$th component $z^{i_{0}}$ of $z\in R^{d}$.
\end{proposition}

\begin{proof}
\textbf{ }The \textquotedblleft if\textquotedblright\ part: Since process
$Y_{t}:=\mathbb{E}_{t,T}[\xi]$ solves the following BSDE
\[
Y_{t}=\xi+\int_{t}^{T}g(s,Y_{s},Z_{s})ds-\int_{t}^{T}Z_{s}dB_{s},
\]
we have
\[
Y_{t}+\int_{0}^{t}\bar{z}_{s}^{i_{0}}dB_{s}^{i_{0}}=\xi+\int_{0}^{T}\bar
{z}_{s}^{i_{0}}dB_{s}^{i_{0}}+\int_{t}^{T}g(s,Y_{s},Z_{s})ds-\int_{t}^{T}%
\bar{Z}_{s}dB_{s},
\]
where
\[
\bar{Z}_{s}=\left(  Z_{s}^{1},\cdots,Z_{s}^{i_{0}-1},Z_{s}^{i_{0}}+\bar{z}%
_{s}^{i_{0}},Z_{s}^{i_{0}+1},\cdots,Z_{s}^{d}\right)  .
\]
But since $g(s,y,z)$ does not depend the $i_{0}$th component of $z\in R^{d}$,
we thus have $g(s,Y_{s},Z_{s})\equiv g(s,Y_{s},\bar{Z}_{s})$. Thus
\[
Y_{t}+\int_{0}^{t}\bar{z}_{s}^{i_{0}}dB_{s}^{i_{0}}=\xi+\int_{0}^{t}\bar
{z}_{s}^{i_{0}}dB_{s}^{i_{0}}+\int_{t}^{T}g(s,Y_{s},\bar{Z}_{s})ds-\int%
_{t}^{T}\bar{Z}_{s}dB_{s}.
\]
This means that (\ref{zBi0}) holds. \newline The \textquotedblleft only
if\textquotedblright\ part: For each fixed $(t,y,z)$, let $(Y_{s}%
^{t,y,z})_{s\geq t}$ be the solution of (\ref{Ytyz}). We have,
\[
\mathbb{E}_{s,T}[Y_{T}^{t,y,z}]-z^{i_{0}}B_{s}^{i_{0}}=\mathbb{E}_{s,T}%
[Y_{T}^{t,y,z}-z^{i_{0}}B_{T}^{i_{0}}],\;s\in\lbrack t,T].
\]
Since the process $Y_{r}:=\mathbb{E}_{s,r}[Y_{r}^{t,y,z}-z_{r}^{i_{0}}%
B_{r}^{i_{0}}]$, $r\in\lbrack t,s]$, solves the BSDE
\[
Y_{s}^{t,y,z}-z^{i_{0}}B_{s}^{i_{0}}=Y_{s}=Y_{T}^{t,y,z}+z^{i_{0}}B_{T}%
^{i_{0}}+\int_{s}^{T}g(r,Y_{r},Z_{r})ds-\int_{s}^{T}Z_{r}dB_{r}.
\]
From which we deduce $Z_{s}=\bar{z}:=\left(  z^{1},\cdots,z^{i_{0}%
-1},0,z^{i_{0}+1},\cdots,z^{d}\right)  =z$ and thus
\[
g(r,Y_{r},Z_{r})=g(r,Y_{r}^{t,y,z},\bar{z})=g(r,Y_{r}^{t,y,z},z),\;0\leq t\leq
r\leq T.
\]
It then follows from Lemma \ref{Lem-f-f-} that
\[
g(t,y,\bar{z})=g(t,y,z),\;t\geq0\text{, a.e., a.s.,}%
\]
i.e., $g$ does not depend the $i_{0}$th component of $z\in R^{d}$.
\end{proof}

\begin{proposition}
The following condition are equivalent: \newline\textbf{(i)} For each $0\leq
s\leq t$, and $X\in L^{2}(\mathcal{F}_{t}^{s})$, the $g$--pricing mechanism
$\mathbb{E}_{s,t}^{g}[X]$ is deterministic; \newline\textbf{(ii) }The
corresponding pricing generating function $g$ is a deterministic function of
$(t,y,z)\in\lbrack0,T]\times R\times R^{d}$. $\;$
\end{proposition}

The proof is similar as the others. We omit it.

\section{Pricing accumulated contingent claim by a general $\mathbb{E}%
_{s,t}[\cdot]$\label{ss4}}

For a given $K\in D_{\mathcal{F}}^{2}(0,\infty)$, we will find the
corresponding definition $\mathbb{E}_{s,t}[\cdot;K]$, for an abstract dynamic
pricing mechanism $\mathbb{E}_{s,t}[\cdot]$ defined on $[0,\infty)$. To this
end we first consider the case $K\in D_{\mathcal{F}}^{0}(0,\infty)$, the space
of step processes defined by%

\begin{equation}
D_{\mathcal{F}}^{0}(0,\infty):=\{K_{t}=\sum_{i=0}^{N-1}\xi_{i}1_{[t_{i}%
,t_{i+1})}(t),\ \{t_{i}\}_{i=0}^{T}\in\pi_{\lbrack0,\infty)}^{N},\ \xi_{i}\in
L^{2}(\mathcal{F}_{t_{i}})\}. \label{e3.2}%
\end{equation}

Now let $K_{t}=\sum_{i=0}^{N-1}\xi_{i}1_{[t_{i},t_{i+1})}(t)$ be fixed. We
observe that, for each $T>0$ and $X\in L^{2}(\mathcal{F}_{T})$, $(X,K)$ is an
accumulated contingent claim with maturity $T$ in a simple way that, at each
time $t_{i}\leq T$, the buyer of the contract $(X,K)$ receives $K_{t_{i}%
}-K_{t_{i}-}$, and, in addition, she or he receives $X$ at the maturity $T$.

We define, for each $0\leq i\leq N-1$, $s$, $t\in\lbrack t_{i},t_{i+1}]$, with
$s\leq t$ and $X\in L^{2}(\mathcal{F}_{t})$,
\[
\mathbb{E}_{s,t}^{i}[X;K]:=\mathbb{E}_{s,t}[X+K_{t}-K_{s}].
\]

\begin{lemma}
\label{m3.2}For each $i=0,1,2,\cdots,N-1$, $\mathbb{E}_{s,t}[\cdot;K]$,
$t_{i}\leq s\leq t\leq t_{i+1}$ is an $\mathcal{F}_{t}$--consistent pricing mechanism.
\end{lemma}

\begin{proof}
\textbf{ }It is easy to check that (A1), (A2) and (A3) holds. We now prove
(A4), i.e., for each $t_{i}\leq s\leq t\leq t_{i+1}$ and $X\in L^{2}%
(\mathcal{F}_{t})$,
\begin{equation}
1_{A}\mathbb{E}_{s,t}^{i}[X;K]=1_{A}\mathbb{E}_{s,t}^{i}[1_{A}X;K],\; \forall
A\in\mathcal{F}_{s}. \label{e3.4}%
\end{equation}
We have
\begin{align*}
1_{A}\mathbb{E}_{s,t}^{i}[X;K_{\cdot}]  &  =1_{A}\mathbb{E}_{s,t}%
[X+K_{t}-K_{s}]\\
\  &  =1_{A}\mathbb{E}_{s,t}[1_{A}(1_{A}X+K_{t}-K_{s})]\\
\  &  =1_{A}\mathbb{E}_{s,t}[1_{A}X+K_{t}-K_{s}]\\
\  &  =1_{A}\mathbb{E}_{s,t}^{i}[1_{A}X;K_{\cdot}].
\end{align*}
Thus (A4) holds.
\end{proof}

By Proposition \ref{m2.3}, there exists a unique $\mathcal{F}_{t}$--consistent
pricing mechanism $\mathbb{E}[\cdot;K]$, that coincides with $\mathbb{E}%
^{i}[\cdot;K]$ for each interval $[t_{i},t_{i+1}]$.

\begin{definition}
\label{m3.1}We denote this unique $\mathcal{F}_{t}$--consistent pricing
mechanism that coincides with $\mathbb{E}^{i}[\cdot;K]$ by $\mathbb{E}%
_{s,t}[\cdot;K]$:
\[
\mathbb{E}_{s,t}[X;K]:X\in L^{2}(\mathcal{F}_{t})\rightarrow L^{2}%
(\mathcal{F}_{s}).
\]

\end{definition}

\begin{remark}
It is easy to check that, for each accumulated contingent claim $(X,K)$ with
maturity $t$ and $K\in D_{\mathcal{F}}^{0}(0,t)$, the only consistent price of
$(X,K)$ of the pricing mechanism $\mathbb{E}$ is $\mathbb{E}_{s,t}[X;K]_{0\leq
s\leq t}$.
\end{remark}

\begin{proposition}
\label{m3.7}We assume that $\mathbb{E}_{s,t}[\cdot]_{0\leq s\leq t<\infty}$ is
a given pricing mechanism satisfying (A1)--(A5) and (A4$_{0}$). Then for each
$K\in D_{\mathcal{F}}^{2}(0,\infty)$ there exists a pricing mechanism
$\mathbb{E}_{s,t}[\cdot;K]_{0\leq s\leq t<\infty}$ which is dominated by
$\mathbb{E}^{g_{\mu}}$ in the following sense, for each $K$, $K^{\prime}\in
D_{\mathcal{F}}^{2}(0,\infty)$ and for each $0\leq s\leq t$, $X$, $X^{\prime
}\in L^{2}(\mathcal{F}_{t})$, we have
\begin{align}
\mathbb{E}_{s,t}^{-g_{\mu}}[X-X^{\prime};(K-K^{\prime})_{\cdot}]  &
\leq\mathbb{E}_{s,t}[X;K_{\cdot}]-\mathbb{E}_{s,t}[X^{\prime};K_{\cdot
}^{\prime}]\label{e3.10}\\
\  &  \leq\mathbb{E}_{s,t}^{g_{\mu}}[X-X^{\prime};(K-K^{\prime})_{\cdot
}]\hbox{, a.s.}\nonumber
\end{align}
and%
\begin{equation}
\mathbb{E}_{s,t}^{-g_{\mu}}[0;K_{\cdot}]\leq\mathbb{E}_{s,t}[0;K_{\cdot}%
]\leq\mathbb{E}_{s,t}^{g_{\mu}}[0;K_{\cdot}] \label{e3.10a}%
\end{equation}
Moreover, for each $K_{t}=\sum_{i=0}^{N-1}\xi_{i}I_{[t_{i},t_{i+1})}(t)$, we
have
\[
\mathbb{E}_{s,t}[X;K_{\cdot}]=\mathbb{E}_{s,t}[X+K_{t}-K_{s}],\ \forall\lbrack
s,t]\in\lbrack t_{i},t_{i+1}]\
\]
Such pricing mechanism is uniquely defined. Furthermore, under the pricing
mechanism $\mathbb{E}$ the price process of the accumulated contingent claim
$(X,K)$ with maturity $t$ is $\mathbb{E}_{s,t}[X;K_{\cdot}]$, $s\leq t$.
\end{proposition}

The prove of this proposition can be found in \cite{Peng2005d}.

\section{$\mathbb{E}[\cdot;K]$--martingales \label{ss5}}

Hereinafter, $\mathbb{E}_{s,t}[\cdot]$ will be a fixed $\mathcal{F}_{t}%
$--consistent pricing mechanism satisfying (A1)--(A5) and (A4$_{0}$). Similar
to $\mathbb{E}_{s,t}^{g}[\cdot]$-pricing mechanism, we introduce the notion of
$\mathbb{E}[\cdot;K]$--martingale:

\begin{definition}
\label{m4.1}Let $K\in D_{\mathcal{F}}^{2}(0,T)$ be given. A process $Y\in
L_{\mathcal{F}}^{2}(t_{0},t_{1})$ satisfying $E[\mathrm{ess}\sup_{s\in
[t_{1},t_{1}]}|Y_{s}|^{2}]<\infty$, is said to be an $\mathbb{E}[\cdot
;K]$--martingales (resp. $\mathbb{E}[\cdot;K]$--supermartingale,
$\mathbb{E}[\cdot;K]$--submartingale) on $[t_{0},t_{1}]$ if for each
$t_{0}\leq s\leq t\leq t_{1}$, we have
\begin{equation}
\mathbb{E}_{s,t}[Y_{t};K]=Y_{s}\hbox{, (resp. }\leq Y_{s}\hbox{,
}\geq Y_{s}\hbox{), a.s.} \label{e4.1}%
\end{equation}

\end{definition}

We then can apply $\mathbb{E}^{g}$--supermartingale decomposition theorem,
i.e., Proposition \ref{p2.3}, to get the following result.

\begin{proposition}
\label{m4.3}We assume (A1)--(A5) and (A4$_{0}$). Let $K$ $\in D_{\mathcal{F}%
}^{2}(0,T)$ be given. For fixed $t\in[0,T]$ and $X\in L^{2}(\mathcal{F}_{t})$,
the process $Y_{s}^{t,X,K}:=\mathbb{E}_{s,t}[X;K],\;s\in[0,t]$, has the
following expression: there exist processes $(g_{\cdot}^{t,X,K},z_{\cdot
}^{t,X,K})\in L_{\mathcal{F}}^{2}(0,t;R\times R^{d})$ such that
\begin{equation}
Y_{s}^{t,X,K}=X+K_{t}-K_{s}+\int_{s}^{t}g_{r}^{t,X,K}dr-\int_{s}^{t}%
z_{r}^{t,X,K}dB_{r},\; \;s\in[0,t], \label{e4.4}%
\end{equation}
such that
\begin{equation}
|g_{s}^{t,X,K}|\leq\mu(|Y_{s}^{t,X,K}|+|z_{s}^{t,X,K}|),\; \forall s\in[0,t].
\label{e4.5}%
\end{equation}
Moreover let $Y_{s}^{t,X^{\prime},K^{\prime}}:=\mathbb{E}_{s,t_{1}}[X^{\prime
};K^{\prime}]$, $s\in[0,t]$, for some other $K^{\prime}\in D_{\mathcal{F}}%
^{2}(0,T)$, $X^{\prime}\in L^{2}(\mathcal{F}_{t})$ and let $(g_{\cdot
}^{t,X^{\prime},K^{\prime}},z_{\cdot}^{t,X^{\prime},K^{\prime}})$ be the
corresponding expression in (\ref{e4.4}), then we have
\begin{equation}
|g_{s}^{t,X,K}-g_{s}^{t,X^{\prime},K^{\prime}}|\leq\mu(|Y_{s}^{t,X,K}%
-Y_{s}^{t,X^{\prime},K^{\prime}}|+|z_{s}^{t,X,K}-z_{s}^{t,X^{\prime}%
,K^{\prime}}|),\; \forall s\in[0,t]. \label{e4.6}%
\end{equation}

\end{proposition}

\begin{proof}
\textbf{ }Since $(Y_{s}^{t,X,K})_{s\in\lbrack0,t]}$, is an $\mathbb{E}%
^{g_{\mu}}[\cdot;K]$--submartingale and $\mathbb{E}^{-g_{\mu}}[\cdot
;K]$--super--martingale, by Proposition \ref{p2.3} and Corollary \ref{c2.2},
there exists an increasing process $A_{\cdot}^{+}\in D_{\mathcal{F}}^{2}(0,t)$
and $A_{\cdot}^{-}\in D_{\mathcal{F}}^{2}(0,t)$ with $A_{0}^{+}=A_{0}^{-}=0$,
such that
\begin{equation}
Y_{s}^{t,X,K}=\mathbb{E}_{s,t}^{g_{\mu}}[X;(K-A^{+})_{\cdot}]=\mathbb{E}%
_{s,t}^{g_{\mu}}[X;(K+A^{-})_{\cdot}],\;s\in\lbrack0,t]. \label{e4.7}%
\end{equation}
According to the notion of $\mathbb{E}^{g}$ defined in (\ref{e2.3}),
$Y_{s}^{t,X,K}$ is the solution of the following BSDE on $[0,t]$:
\begin{align}
Y_{s}^{t,X,K}  &  =X+(K-A^{+})_{t}-(K-A^{+})_{s}\label{e4.8}\\
&  \ +\int_{s}^{t}\mu(|Y_{r}^{t,X,K}|+|Z_{r}^{+}|)dr-\int_{s}^{t}Z_{r}%
^{+}dB_{r}\nonumber
\end{align}
and
\begin{align}
Y_{s}^{t,X,K}  &  =X+(K+A^{-})_{t}-(K+A^{-})_{s}\label{e4.9}\\
&  -\int_{s}^{t}\mu(|Y_{r}^{t,X,K}|+|Z_{r}^{-}|)dr-\int_{s}^{t}Z_{r}^{-}%
dB_{r}.\nonumber
\end{align}
It then follows that $Z_{s}^{t,X,K}:=Z_{s}^{+}\equiv Z_{s}^{-}$, $s\in
\lbrack0,t]$ and thus
\[
-dA_{s}^{+}+\mu(|Y_{s}^{t,X,K}|+|Z_{s}^{t,X,K}|)ds\equiv dA_{s}^{-}-\mu
(|Y_{s}^{t,X,K}|+|Z_{s}^{t,X,K}|)ds,
\]
or
\begin{equation}
dA_{s}^{-}+dA_{s}^{+}\equiv2\mu(|Y_{s}^{t,X,K}|+|Z_{s}^{t,X,K}|)ds,\;s\in
\lbrack0,t] \label{e4.10}%
\end{equation}
Thus $dA^{+}$ and $dA^{-}$ are absolutely continuous with respect to $ds$. We
denote $a_{s}^{+}ds=dA_{s}^{+}$ and $a_{s}^{-}ds=dA_{s}^{-}$. It is clear
that
\begin{align*}
0  &  \leq a_{s}^{+}\leq2\mu(|Y_{s}^{t,X,K}|+|Z_{s}^{t,X,K}|),\\
0  &  \leq a_{s}^{-}\leq2\mu(|Y_{s}^{t,X,K}|+|Z_{s}^{t,X,K}|),\;dP\times
dt\hbox{--a.e.}
\end{align*}
We then can rewrite (\ref{e4.8}) as
\begin{equation}
Y_{s}^{t,X,K}=X+K_{t}-K_{s}+\int_{s}^{t}[-a_{r}^{+}+\mu(|Y_{r}^{t,X,K}%
|+|Z_{r}^{+}|)]dr-\int_{s}^{t}Z_{r}^{+}dB_{r}. \label{e4.11}%
\end{equation}
Thus, by setting $g_{r}^{t,X,K}:=-a_{r}^{+}+\mu(|Y_{r}^{t,X,K}|+|Z_{r}^{+}|)$,
we have the expression (\ref{e4.4}) as well as the estimate (\ref{e4.5}).
\end{proof}

It remains to prove (\ref{e4.6}). By (A5) of Proposition \ref{p2.3} $\hat
{Y}_{s}=Y_{s}^{t,X,K}-Y_{s}^{t,X^{\prime},K^{\prime}}$ is an $\mathbb{E}%
^{g_{\mu}}[\cdot;K-K^{\prime}]$--submartingale and an $\mathbb{E}^{-g_{\mu}%
}[\cdot;K-K^{\prime}]$--supermartingale on $[0,t]$. Thus we can repeat the
above procedure to prove that there exist processes $(\hat{g}_{\cdot},\hat
{Z}_{\cdot})\in L_{\mathcal{F}}^{2}(0,t;R\times R^{d})$ such that
\begin{equation}
\hat{Y}_{s}=X-X^{\prime}+(K-K^{\prime})_{t}-(K-K^{\prime})_{s}+\int_{s}%
^{t}\hat{g}_{r}dr-\int_{s}^{t}\hat{Z}_{r}dB_{r},\; \;s\in\lbrack0,t],
\label{e4.12}%
\end{equation}
such that
\begin{equation}
|\hat{g}_{s}|\leq\mu(|\hat{Y}_{s}|+|\hat{Z}_{s}|),\; \forall s\in\lbrack0,t].
\label{e4.13}%
\end{equation}
But by (\ref{e4.4}) and $\hat{Y}_{s}\equiv Y_{s}^{t,X,K}-Y_{s}^{t,X^{\prime
},K^{\prime}}$, we immediately have
\begin{equation}
\hat{g}_{s}\equiv g_{s}^{t,X,K}-g_{s}^{t,X^{\prime},K^{\prime}},\; \hat{Z}%
_{s}\equiv z_{s}^{t,X,K}-z_{s}^{t,X^{\prime},K^{\prime}}. \label{e4.14}%
\end{equation}
This with (\ref{e4.13}) yields (\ref{e4.6}). The proof is complete.\

\begin{corollary}
\label{m4.4}Let $K^{1}$ and $K^{2}\in D_{\mathcal{F}}^{2}(0,T)$ and $X^{1}\in
L^{2}(\mathcal{F}_{t_{1}})$, $X^{2}\in L^{2}(\mathcal{F}_{t_{2}})$ be given
for some fixed $0\leq t_{1}\leq t_{2}\leq T$ and let $(g_{s}^{t_{i}%
,X^{i},K^{i}},Z_{s}^{t_{i},X^{i},K^{i}})_{s\in[0,t_{i}]}$, $i=1,2$, be the
pair in (\ref{e4.4}) for $Y_{s}^{t_{i},X^{i},K^{i}}=\mathbb{E}_{s,t_{i}}%
[X^{i};(K^{i})_{\cdot}]$, $i=1,2$, respectively. Then we have
\begin{equation}
|g_{s}^{t,X^{1},K^{1}}-g_{s}^{t,X^{2},K^{2}}|\leq\mu(|Y_{s}^{t_{1},X^{1}%
,K^{2}}-Y_{s}^{t_{2},X^{2},K^{2}}|+|z_{s}^{t_{1},X^{1},K^{1}}-z_{s}%
^{t_{2},X^{2},K^{2}}|),\; \forall s\in[0,t_{1}]. \label{e4.15}%
\end{equation}

\end{corollary}

\begin{proof}
With the observation
\[
Y_{s}^{t_{2},X^{2},K^{2}}=\mathbb{E}_{s,t_{1}}[Y_{t_{1}}^{t_{2},X^{2},K^{2}%
};(K^{2})_{\cdot}],\;s\in\lbrack0,t_{1}],
\]
it is an immediate consequence of Proposition \ref{m4.3}.\
\end{proof}

\begin{corollary}
\label{m4.5}For each $t\in[0,T]$ and $X\in L^{2}(\mathcal{F}_{t})$, $K\in
D_{\mathcal{F}}^{2}(0,T)$, the process $(\mathbb{E}_{s,t}[X;K_{\cdot}%
])_{s\in[0,t]}$ is also in $D_{\mathcal{F}}^{2}(0,t)$. If moreover, $K\in
S_{\mathcal{F}}^{2}(0,T)$ (resp. It\^o's process), then $(\mathbb{E}%
_{s,t}[X;K_{\cdot}])_{s\in[0,t]}$ is also in $S_{\mathcal{F}}^{2}(0,t)$ (resp.
It\^o's process).
\end{corollary}

\section{$\mathbb{E}$--supermartingale decomposition theorem: intrinsic
formulation\label{ss7}}

Our objective of this section is to give the following $\mathbb{E}%
$--supermartingale decomposition theorem of Doob--Meyer's type. Since
$(\mathbb{E}_{s,t}[\cdot])_{s\leq t}$ is abstract and nonlinear, it is
necessary to introduce the intrinsic form (\ref{e6.1}). This theorem plays an
important role in the proof of the main theorem of this paper. It can be also
considered as a generalization of Proposition \ref{p2.3}. This is a very
profond theorem, the proof can be found in \cite[Peng 2005]{Peng2005d}.

\begin{theorem}
\label{m6.1}We assume (A1)--(A5) as well as (A4$_{0}$). Let $Y\in
S_{\mathcal{F}}^{2}(0,T)$ be an $\mathbb{E}[\cdot]$--supermartingale. Then
there exists an increasing process $A\in S_{\mathcal{F}}^{2}(0,T)$ with
$A_{0}=0$, such that $Y$ is an $\mathbb{E}[\cdot;A]$--martingale, i.e.,
\begin{equation}
Y_{t}=\mathbb{E}_{t,T}[Y_{T};A_{\cdot}],\;t\in[0,T]. \label{e6.1}%
\end{equation}

\end{theorem}

\begin{remark}
\label{m6.1Rem1}This theorem has an interesting interpretation: the fact that
$Y\in S_{\mathcal{F}}^{2}(0,T)$ is an $\mathbb{E}[\cdot]$--supermartingale
means that if $Y$ is always undervalued by the pricing mechanism $\mathbb{E}$,
i.e., $\mathbb{E}_{t,T}[Y_{T}]\leq Y_{t}$, then there exists an increasing
process $A$ such that $Y_{t}$ is just the $\mathbb{E}$ price of the
accumulated contingent claim $(Y_{T},A)$ at maturity $T$.
\end{remark}

\begin{remark}
\label{m6.1Rem2}In the case where $(\mathbb{E}_{s,t}[\cdot])_{0\leq s\leq
t\leq T}$ is a system of linear mappings, (\ref{e6.1}) becomes
\[
Y_{t}+A_{t}=\mathbb{E}_{t,T}[Y_{T}+A_{T}],\;t\in[0,T],
\]
i.e., as in classical situation, $Y+A$ is an $\mathbb{E}[\cdot]$--martingale.
But, the intrinsic formulation that can be applied to nonlinear situation is
that $Y$ is an $\mathbb{E}[\cdot;A]$--martingale.
\end{remark}

\section{Appendix \label{ss8}}

\subsection{Proof of Theorem \ref{m7.1}}

For each fixed $(t,y,z)\in\lbrack0,T]\times R\times R^{d}$, we consider the
solution $Y^{t,y,z}\in S_{\mathcal{F}}^{2}(0,T)$ of a It\^{o}'s equation on
$[t,T]$:
\begin{align}
dY_{s}^{t,y,z}  &  =-g_{\mu}(Y_{s}^{t,y,z},z)ds+zdB_{s},\; \;s\in
(t,T],\label{e7.2}\\
Y_{t}^{t,y,z}  &  =y. \label{e7.4}%
\end{align}
It is easy to check that $Y^{t,y,z}$ is an $\mathbb{E}^{g_{\mu}}[\cdot
]$--martingale, i.e., it is a price process of the pricing mechanism
$\mathbb{E}^{g_{\mu}}[\cdot]$ on $[t,T]$. Since the pricing mechanism
$\mathbb{E}[\cdot]$ is dominated by $\mathbb{E}^{g_{\mu}}[\cdot]$, from
(\ref{e3.1}) $Y^{t,y,z}$ is also an $\mathbb{E}[\cdot]$--supermartingale. By
Decomposition Theorem \ref{m6.1}, there exists an increasing process
$A^{t,y,z}\in S_{\mathcal{F}}^{2}(0,T)$ with $A_{0}^{t,y,z}=0$, such that
\begin{equation}
Y_{s}^{t,y,z}=\mathbb{E}_{s,T}[Y_{T}^{t,y,z};A_{\cdot}^{t,y,z}]. \label{e7.5}%
\end{equation}
i.e., $Y_{\cdot}^{t,y,z}$ is just the pricing process produced by
$\mathbb{E}[\cdot]$ of the accumulated contingent claim $(Y_{T}^{t,y,z}%
;A_{\cdot}^{t,y,z})$ with maturity $T$. By Proposition \ref{m4.3} and
Corollary \ref{m4.4}, there exists $($ $g^{t,y,z},Z^{t,y,z})\in L_{\mathcal{F}%
}^{2}(0,T)$ such that
\begin{equation}
-dY_{s}^{t,y,z}=dA_{s}^{t,y,z}+g_{s}^{t,y,z}ds-Z_{s}^{t,y,z}dB_{s}%
,\;s\in\lbrack t,T], \label{e7.6}%
\end{equation}
and such that, for each different $(t,y,z)$, $(t^{\prime},y^{\prime}%
,z^{\prime})$ $\in\lbrack0,T]\times R\times R^{d}$%
\begin{equation}
|g_{s}^{t,y,z}-g_{s}^{t^{\prime},y^{\prime},z^{\prime}}|\leq\mu|Y_{s}%
^{t,y,z}-Y_{s}^{t^{\prime},y^{\prime},z^{\prime}}|+\mu|Z_{s}^{t,y,z}%
-Z_{s}^{t^{\prime},y^{\prime},z^{\prime}}|,\;s\in\lbrack t\vee t^{\prime},T],
\label{e7.7}%
\end{equation}
and
\begin{equation}
|g_{s}^{t,y,z}|\leq\mu|Y_{s}^{t,y,z}|+\mu|Z_{s}^{t,y,z}|,\;s\in\lbrack
t,T],\;ds\times dP\hbox{--a.e.} \label{e7.7a}%
\end{equation}

Now for each $X\in L^{2}(\mathcal{F}_{t^{\prime}})$, we set
\begin{equation}
\bar Y_{s}^{t^{\prime},X}:=\mathbb{E}_{s,t^{\prime}}[X]=\mathbb{E}%
_{s,t^{\prime}}[X;0]. \label{e7.8a}%
\end{equation}
We use once more Proposition \ref{m4.3} and Corollary \ref{m4.4}: there exists
$($ $\bar g^{t^{\prime},X},\bar Z^{t^{\prime},X})\in L_{\mathcal{F}}%
^{2}(0,t^{\prime})$ such that, for $s\in[0,t^{\prime}]$,
\begin{equation}
-d\bar Y_{s}^{t^{\prime},X}=\bar g_{s}^{t^{\prime},X}ds-\bar Z_{s}^{t^{\prime
},X}dB_{s},\; \bar Y_{t^{\prime}}=X, \label{e7.8}%
\end{equation}
such that
\begin{equation}
|g_{s}^{t,y,z}-\bar g_{s}^{t^{\prime},X}|\leq\mu|Y_{s}^{t,y,z}-\bar
Y_{s}^{t^{\prime},X}|+\mu|Z_{s}^{t,y,z}-\bar Z_{s}^{t^{\prime},X}%
|,\;s\in[t,t^{\prime}],\;ds\times dP\hbox{--a.e.} \label{e7.9}%
\end{equation}
and, for $X$, $X\in L^{2}(\mathcal{F}_{t^{\prime}})$,
\[
|\bar g_{s}^{t^{\prime},X}-\bar g_{s}^{t^{\prime},X^{\prime}}|\leq\mu|\bar
Y_{s}^{t^{\prime},X}-\bar Y_{s}^{t^{\prime},X^{\prime}}|+\mu|\bar
Z_{s}^{t^{\prime},X}-\bar Z_{s}^{t^{\prime},X^{\prime}}|,\;s\in[0,t^{\prime
}],\;ds\times dP\hbox{--a.e..} \label{e7.10}
\]
On the other hand, comparing to (\ref{e7.2}) and (\ref{e7.6}), we have
\[
Z_{s}^{t,y,z}\equiv1_{[t,T]}(s)z. \label{e7.11}
\]
Thus (\ref{e7.7}), (\ref{e7.7a}) and (\ref{e7.9}) become, respectively,
\begin{equation}
|g_{s}^{t,y,z}-g_{s}^{t^{\prime},y^{\prime},z^{\prime}}|\leq\mu|Y_{s}%
^{t,y,z}-Y_{s}^{t^{\prime},y^{\prime},z^{\prime}}|+\mu|z-z^{\prime}%
|,\;s\in[t\vee t^{\prime},T],\;ds\times dP\hbox{--a.e.,} \label{e7.12}%
\end{equation}
\begin{equation}
|g_{s}^{t,y,z}|\leq\mu|Y_{s}^{t,y,z}|+\mu|z|,\; \label{e7.12a}%
\end{equation}
and
\begin{equation}
|g_{s}^{t,y,z}-\bar g_{s}^{t^{\prime},X}|\leq\mu|Y_{s}^{t,y,z}-\bar
Y_{s}^{t^{\prime},X}|+\mu|z-\bar Z_{s}^{t^{\prime},X}|,\;s\in[t,t^{\prime
}],\;ds\times dP\hbox{--a.e.} \label{e7.13}%
\end{equation}

Now, for each $n=1,2,3,\cdots$, we set $t_{i}^{n}=i2^{-n}T$, $i=0,1,2,\cdots
,2^{n}$, and define
\begin{equation}
g^{n}(s,y,z):=\sum_{i=0}^{2^{n}-1}g_{s}^{t_{i}^{n},y,z}1_{[t_{i}^{n}%
,t_{i+1}^{n})}(s),\;s\in\lbrack0,T],\ (y,z)\in R\times R^{d}. \label{e7.14}%
\end{equation}
It is clear that $g^{n}$ is an $\mathcal{F}_{t}$--adapted process.

\begin{lemma}
\label{m7.2aa}For each fixed $(y,z)\in R\times R^{d}$ and $T>0$,
$\{g^{n}(\cdot,y,z)\}_{n=1}^{\infty}$ is a Cauchy sequence in $L_{\mathcal{F}%
}^{2}(0,T)$.
\end{lemma}

\begin{proof}
Let $0<m<n$ be two integers. For each $s\in\lbrack0,T)$, there are some
integers $i\leq2^{m}-1$ and $j\leq2^{n}-1$ with $t_{i}^{m}\leq t_{j}^{m}$,
such that $s\in\lbrack t_{i}^{m},t_{i+1}^{m})\cap\lbrack t_{j}^{n},t_{j+1}%
^{n})$. We have, by (\ref{e7.12})
\begin{align*}
|g^{m}(s,y,z)-g^{n}(s,y,z)|  &  =|g_{s}^{t_{i}^{m},y,z}-g_{s}^{t_{j}^{n}%
,y,z}|\\
&  \leq\mu|Y_{s}^{t_{i}^{m},y,z}-Y_{s}^{t_{j}^{n},y,z}|\\
&  \leq\mu|Y_{s}^{t_{i}^{m},y,z}-y|+\mu|Y_{s}^{t_{j}^{n},y,z}-y|.
\end{align*}
By (\ref{ee7.15}) of Lemma \ref{em7.2a} given later,
\begin{align*}
&  E[|g^{m}(s,y,z)-g^{n}(s,y,z)|^{2}]\\
&  \leq2\mu^{2}C(|y|^{2}+|z|^{2}+1)(2^{-m}+2^{-n})T.
\end{align*}
Thus
\begin{align}
\sup_{s\in\lbrack0,T)}E[|g^{m}(s,y,z)-g^{n}(s,y,z)|^{2}]  &  \leq2\mu
^{2}E[|Y_{s}^{t_{i}^{m},y,z}-y|^{2}+|Y_{s}^{t_{j}^{n},y,z}-y|^{2}%
]\label{e7.16}\\
&  \leq2\mu^{2}C(|y|^{2}+|z|^{2}+1)(2^{-m}+2^{-n})T.\nonumber
\end{align}
Thus $\{g^{n}(\cdot,y,z)\}_{n=1}^{\infty}$ is a Cauchy sequence in
$L_{\mathcal{F}}^{2}(0,T)$.\
\end{proof}

We can give

\begin{definition}
\label{m7.3}For each $(y,z)\in R\times R^{d}$, we denote $g(\cdot,y,z)\in
L_{\mathcal{F}}^{2}(0,T)$, the Cauchy limit of $\{g^{n}(\cdot,y,z)\}_{n=1}%
^{\infty}$ in $L_{\mathcal{F}}^{2}(0,T)$.
\end{definition}

We will prove that the pricing mechanism $\mathbb{E}_{s,t}[\cdot]$ is just the
$g$--pricing mechanism with $g$ obtained in the above definition as its
generating function, and thus our main result Theorem \ref{m7.1} hold true. We
still need to investigate some important properties of $g$. We have the
following estimates for the function $g$.

\begin{lemma}
\label{m7.4}The limit $g:\Omega\times\lbrack0,T]\times R\times R^{d}%
\rightarrow R^{d}$ satisfies the following properties:
\begin{equation}
\left\{
\begin{array}
[c]{rrl}%
\text{(i)}\, &  & g(\cdot,y,z)\in L_{\mathcal{F}}^{2}%
(0,T)\hbox{, for each }(y,z)\in R\times R^{d};\\
\text{(ii)}\, &  & |g(s,y,z)-g(s,y^{\prime},z^{\prime})|\leq\mu(|y-y^{\prime
}|+|z-z^{\prime}|),\; \forall y,y^{\prime}\in R,\;z,z^{\prime}\in R^{d};\\
\text{(iii)}\, &  & g(s,0,0)\equiv0;\\
\text{(iv)}\, &  & |g(s,y,z)-\bar{g}^{t,X}|\leq\mu|y-\bar{Y}_{s}^{t,X}%
|+\mu|z-\bar{Z}_{s}^{t,X}|,\forall s\in\lbrack0,t],\ X\in L^{2}(\mathcal{F}%
_{t}).
\end{array}
\right.  \label{e7.17}%
\end{equation}
where $(\bar{Y}^{t,X},\bar{Z}^{t,X})$ is the process defined in (\ref{e7.8a})
and (\ref{e7.8}).
\end{lemma}

\begin{proof}
\textbf{ }(i) is clear. To prove (ii), we choose $t_{i}^{n}=i2^{-n}T$,
$i=0,1,2,\cdots,2^{n}$ as in (\ref{e7.14}). For each $s\in\lbrack0,T)$. We
have, once more by (\ref{e7.12}),
\begin{align}
|g^{n}(s,y,z)-g^{n}(s,y^{\prime},z^{\prime})|  &  =\sum_{j=0}^{2^{n}%
-1}1_{[t_{j}^{n},t_{j+1}^{n})}(s)|g_{s}^{t_{j}^{n},y,z}-g_{s}^{t_{j}^{n}%
,y,z}|\label{e7.17a}\\
&  \leq\mu\sum_{j=0}^{2^{n}-1}1_{[t_{j}^{n},t_{j+1}^{n})}(s)(|Y_{s}^{t_{j}%
^{n},y,z}-Y_{s}^{t_{j}^{n},y,z}|+|z-z^{\prime}|)\nonumber\\
&  \leq\mu\sum_{j=0}^{2^{n}-1}1_{[t_{j}^{n},t_{j+1}^{n})}(s)(|Y_{s}^{t_{j}%
^{n},y,z}-y|+|Y_{s}^{t_{j}^{n},y,z}-y^{\prime}|)\nonumber\\
&  +\mu(|y-y^{\prime}|+|z-z^{\prime}|)\nonumber
\end{align}
The first term $I^{n}(s)$ of the right hand is dominated by, using
(\ref{ee7.15}),
\begin{align*}
E[|I^{n}(s)|^{2}]  &  \leq2\mu^{2}\sum_{i=0}^{2^{n}-1}1_{[t_{j}^{n}%
,t_{j+1}^{n})}(s)E[|Y_{s}^{t_{j}^{n},y,z}-y|^{2}+|Y_{s}^{t_{j}^{n},y^{\prime
},z^{\prime}}-y^{\prime}|^{2}]\\
&  \leq2\mu^{2}\sum_{i=0}^{2^{n}-1}1_{[t_{j}^{n},t_{j+1}^{n})}(s)C(|y|^{2}%
+|z|^{2}+|y^{\prime}|^{2}+|z^{\prime}|^{2}+2)2^{-n}T.
\end{align*}
Thus $I^{n}(\cdot)\rightarrow0$ in $L_{\mathcal{F}}^{2}(0,T)$ as
$n\rightarrow\infty$. (ii) is obtained by passing to the limit in both sides
of (\ref{e7.17a}). (iii) is proved similarly by using (\ref{e7.12a}) and
(\ref{ee7.15}). \newline To prove (iv), We apply (\ref{e7.13}),
\end{proof}

\begin{align*}
|g^{n}(s,y,z)-\bar{g}_{s}^{t,X}|  &  =\sum_{i=0}^{2^{n}-1}1_{[t_{j}%
^{n},t_{j+1}^{n})}(s)|g_{s}^{t_{j}^{n},y,z}-\bar{g}_{s}^{t,X}|\\
\  &  \leq\sum_{i=0}^{2^{n}-1}1_{[t_{j}^{n},t_{j+1}^{n})}(s)[\mu|Y_{s}%
^{t_{j}^{n},y,z}-\bar{Y}_{s}^{t,X}|+\mu|z-\bar{Z}_{s}^{t,X}|]\\
\  &  \leq\mu\sum_{i=0}^{2^{n}-1}1_{[t_{j}^{n},t_{j+1}^{n})}(s)|Y_{s}%
^{t_{j}^{n},y,z}-y|+\mu|y-\bar{Y}_{s}^{t,X}|+\mu|z-\bar{Z}_{s}^{t,X}|.
\end{align*}
Then we pass to the limit on both sides.

Finally, We give

\medskip

\noindent\textbf{Proof of Theorem \ref{m7.1}.} For each fixed $t$ and $X\in
L^{2}(\mathcal{F}_{t})$, we denote by $\bar{Y}_{s}^{t,X}:=\mathbb{E}_{s,t}%
[X]$, the $\mathbb{E}$-price on $s\in\lbrack0,t]$ with the contingent $\bar
{Y}_{t}^{t,X}=X$ at the maturity $X$. By Proposition \ref{m4.3} and Corollary
\ref{m4.4}, this price process can has the form
\[
\bar{Y}_{s}^{t,X}=X+\int_{s}^{t}\bar{g}_{r}^{t,X}dr-\int_{s}^{t}\bar{Z}%
_{r}^{t,X}dB_{r},\;s\in\lbrack0,t].
\]
On the other hand, let $Y_{s}^{t,X}=\mathbb{E}_{s,t}^{g}[X]$, the price
process of the contingent claim $X$ generated by $g$. It solves the BSDE
\[
Y_{s}^{t,X}=X+\int_{s}^{t}g(r,Y_{r}^{t,X},Z_{r}^{t,X})dr-\int_{s}^{t}%
Z_{r}^{t,X}dB_{r},\;s\in\lbrack0,t].
\]
By Lemma \ref{m7.4}--(i) and (ii), this BSDE is well--posed. We then apply
It\^{o}'s formula to $|\bar{Y}^{t,X}-Y|^{2}$ on the pricing interval $[0,t]$,
take expectation. Exactly as the classical proof of the uniqueness of BSDE, we
have, using (iv) of Lemma \ref{m7.4}.
\begin{align*}
E|\bar{Y}_{s}^{t,X}-Y_{s}|^{2}  &  +E\int_{s}^{t}|\bar{Z}_{r}^{t,X}-Z_{r}%
|^{2}dEr\\
&  =2E\int_{s}^{t}(\bar{Y}_{r}^{t,X}-Y_{r})(\bar{g}_{r}^{t,X}-g(r,Y_{r}%
,Z_{r}))dr\\
&  \leq2E\int_{s}^{t}(|\bar{Y}_{r}^{t,X}-Y_{r}|\cdot|\bar{g}_{r}%
^{t,X}-g(r,Y_{r},Z_{r})|)dr\\
&  \leq2E\int_{s}^{t}|\bar{Y}_{r}^{t,X}-Y_{r}|\cdot\mu(|\bar{Y}_{r}%
^{t,X}-Y_{r}|+|\bar{Z}_{r}^{t,X}-Z_{r}|)dr\\
&  \leq E\int_{s}^{t}2(\mu+\mu^{2})|\bar{Y}_{r}^{t,X}-Y_{r}|^{2}+\frac{1}%
{2}|\bar{Z}_{r}^{t,X}-Z_{r}|^{2})dr.
\end{align*}
It then follows by using Gronwall's inequality that
\[
\bar{Y}_{s}^{t,X}\equiv Y_{s}=\mathbb{E}_{s,t}^{g}[X],\ \forall s\in
\lbrack0,t].
\]
We thus have the desired result. The proof is complete. $\Box$

\subsection{Testing condition of domination (A5)}

{\Large \bigskip with computational realization by CHEN Lifeng and SUN Peng}

With Chen and Sun of our research group, we have applied our main result
\ref{m7.1} to test if a specific pricing mechanism is a $g$--expectation, or
$g$--pricing mechanism. We need to verify by testing if the crucial assumption
(A5), i.e., the domination inequality (\ref{e3.1}) holds true. We have tested
by using the data of prices of different options given by the price mechanism.

We first test the CME (Chicago Mercantile Exchange)'s market price mechanism
of the options with S\&P500 futures as the underlying asset. The data of the
call and put priceses, from\ year 2000 to 2003, and the corresponding S\&P500
future's prices is obtained from the parameter files of SPAN (Standard
Portfolio Analysis of Risk) system downloaded from CME's ftp site.

We denote by $X_{T}^{i}=(S_{T}-k_{i})^{+}$ (resp. $Y_{T}^{i}=(S_{T}-k_{i}%
)^{-}$, the market price of the call (resp. put ) option with muturity $T$ and
strike price $k_{i}$. We denote their market price at time $t<T$ by
$\mathbb{E}_{t,T}^{m}[X_{T}^{i}]$ and $\mathbb{E}_{t,T}^{m}[Y_{T}^{i}]$,
respectively. The inequalities we need to put to the test are (\ref{e3.1}) in
the following different conbinations, with different $(t,T)$ and different
strike prices%
\begin{equation}
\left\{
\begin{array}
[c]{ccc}%
\text{Call--Call:} &  & \mathbb{E}_{t,T}^{m}[X_{T}^{i}]-\mathbb{E}_{t,T}%
^{m}[X_{T}^{j}]\leq\mathbb{E}_{t,T}^{g_{\mu}}[X_{T}^{i}-X_{T}^{j}]\\
\text{Put--Put:} &  & \mathbb{E}_{t,T}^{m}[Y_{T}^{i}]-\mathbb{E}_{t,T}%
^{m}[Y_{T}^{j}]\leq\mathbb{E}_{t,T}^{g_{\mu}}[Y_{T}^{i}-Y_{T}^{j}]\\
\text{Call--Put:} &  & \mathbb{E}_{t,T}^{m}[X_{T}^{i}]-\mathbb{E}_{t,T}%
^{m}[Y_{T}^{j}]\leq\mathbb{E}_{t,T}^{g_{\mu}}[X_{T}^{i}-Y_{T}^{j}]\\
\text{Put--Call:} &  & \mathbb{E}_{t,T}^{m}[Y_{T}^{i}]-\mathbb{E}_{t,T}%
^{m}[X_{T}^{j}]\leq\mathbb{E}_{t,T}^{g_{\mu}}[Y_{T}^{i}-X_{T}^{j}]
\end{array}
\right.  \label{testedIQ}%
\end{equation}
\ \

In the above inequalities the left hand are market data taken from CME
parameter files. The right hand are the corresponding $g_{\mu}$--expectations.
We have calculated all these values by using standard binomial tree algorithm
of BSDE. Here use the algorithms in Peng and Xu [2005] to solve the following
1-dimensional BSDE:%

\begin{align}
y_{t}  &  =\xi+\int_{t}^{T}\mu(|y_{s}|+|z_{s}|)ds-\int_{t}^{T}z_{s}%
dB_{s}\label{Numerical}\\
y_{T}  &  =X_{T}^{i}-X_{T}^{j}\ \ \text{(resp.\ }Y_{T}^{i}-Y_{T}^{j}%
\text{,\ }X_{T}^{i}-Y_{T}^{j}\text{ and }Y_{T}^{i}-X_{T}^{j}\text{).\ }%
\nonumber
\end{align}

5 parameter files from year 2000 to 2003 have been put in the test. We list
the number of tested inequalities (\ref{testedIQ}) corresponding to each CME
parameter file:%

\begin{tabular}
[t]{|c|c|c|}\hline
CME parameter file name & year & number of tested inequalities\\\hline
cme0105s.par & 2000 & 54584\\\hline
cme0105s.par & 2001 & 62424\\\hline
cme0104s.par & 2002 & 35830\\\hline
cme0103s.par & 2003 & 28162\\\hline
cme0701s.par & 2003 & 61438\\\hline
total number tested &  & 242438\\\hline
\end{tabular}

This means that BSDE (\ref{Numerical}) have been caculated 242438 times (with
CPU P4 Xeron 2.8G). A surpricingly positive result was obtained: among the
totally 242438 tested inqualities, only 5 are against the criteria
(\ref{testedIQ}). Moreover, those 5 counterexamples are singular situation
since they themself all violate Axiomatic monotonicity condition (A1). 5 cases
are all from cme0701s.par, 2003, Put--Put. They are all the singular cases of
form
\[
\mathbb{E}_{t,T}^{m}[(S_{T}-k_{i})^{-}]>\mathbb{E}_{t,T}^{m}[(S_{T}-k_{j}%
)^{-}],\ \ \text{for }k_{i}>k_{j}.
\]
More specific results of the test will be given in our forthcoming paper.

Another feature of our test is, usualy, the bigger $T-t$ is, the
$\mathbb{E}_{t,T}^{g_{\mu}}[X_{T}^{i}-X_{T}^{j}]-(\mathbb{E}_{t,T}^{m}%
[X_{T}^{i}]-\mathbb{E}_{t,T}^{m}[X_{T}^{j}])$. We present 4 features for a
relatively smaller $T-t$ to show the tested result.

\bigskip\medskip Test of inequalities: CME file: cme0105s.par, 2001, $t$: Jan.
05,\ $T$: Jan. 17 for S\&P500 01-03 future with $S_{t}=\$1413.5$\bigskip\

1.\qquad Call--Call:$\ \ \ \ \ \mathbb{E}_{t,T}^{g_{\mu}}[X_{T}^{i}-X_{T}%
^{j}]-(\mathbb{E}_{t,T}^{m}[X_{T}^{i}]-\mathbb{E}_{t,T}^{m}[X_{T}^{j}])>0$

\medskip

2. \ Put--Put:\ \ $\mathbb{E}_{t,T}^{g_{\mu}}[Y_{T}^{i}-Y_{T}^{j}%
]-(\mathbb{E}_{t,T}^{m}[Y_{T}^{i}]-\mathbb{E}_{t,T}^{m}[Y_{T}^{j}])>0$

3.\qquad Call--Put:\ \ \ $\mathbb{E}_{t,T}^{g_{\mu}}[X_{T}^{i}-Y_{T}%
^{j}]-(\mathbb{E}_{t,T}^{m}[X_{T}^{i}]-\mathbb{E}_{t,T}^{m}[Y_{T}^{j}])>0$

\medskip

4.\qquad Put--Call\ $\mathbb{E}_{t,T}^{g_{\mu}}[Y_{T}^{i}-X_{T}^{j}%
]-(\mathbb{E}_{t,T}^{m}[Y_{T}^{i}]-\mathbb{E}_{t,T}^{m}[X_{T}^{j}])>0$

%
%
%
%
%
%
%
%
%
%

\bigskip

Remark: the computations was realized by CHEN Lifeng and SUN Peng

\end{document}